\documentclass[11pt]{article}
\usepackage{amsfonts}
\usepackage{graphics,amssymb,amsthm,amsmath,epsf}
\usepackage{graphicx}
\usepackage{subfigure}

\usepackage[round]{natbib}

\newtheorem{theorem}{Theorem}[section]
\newtheorem{prop}{Proposition}[section]
\newtheorem{example}{Example}[section]
\numberwithin{equation}{section}

\allowdisplaybreaks

\begin{document}
\title{The McDonald Gompertz Distribution: Properties and Applications}
\author{Rasool Roozegar\thanks{Corresponding: rroozegar@yazd.ac.ir },  Saeid Tahmasebi$^\dag$, and Ali Akbar Jafari$^*$\\
{\small $^*$Department of Statistics, Yazd University, Yazd,  Iran}\\
{\small $^\dag$Department of Statistics, Persian Gulf University, Bushehr, Iran}}
\date{}
\maketitle \thispagestyle{empty}

\begin{abstract}
This paper introduces a five-parameter lifetime model with increasing, decreasing, upside -down bathtub and bathtub shaped failure rate called as the McDonald Gompertz (McG) distribution. This new distribution extend the Gompertz, generalized Gompertz, generalized exponential, beta Gompertz and Kumaraswamy Gompertz distributions, among several other models.
We obtain several properties of the McG distribution including moments, entropies, quantile and generating functions. We provide the density function of the order statistics and their moments. The parameter estimation is based on the usual maximum likelihood approach. We also provide the observed information matrix and discuss inferences issues. In the end, the flexibility and usefulness of the new distribution is illustrated by means of  application to two real data sets.
\end{abstract}

\section{Introduction}
The Gompertz (G) distribution, generalizing exponential (E) distribution, is a popular distribution that has been commonly used in many applied problems for modeling data in biology
\cite{economos-82},
gerontology
\cite{br-fo-74},
engineering and marketing studies
\cite{be-gl-12}.
A significant progress has been made towards the generalization and construction flexible distributions to facilitate better modeling of well-known lifetime data. The book by
\cite{jo-ko-ba-95-2}
provides some applications of the G distribution.

In recent years, many authors have proposed distributions which can arise as special submodels within the McDonald generated or generalized beta generated (GBG) class of distributions.
\cite{al-co-or-sa-12}
introduce a class of generalized beta-generated distributions that have three shape parameters in the generator. They considered eleven different parents: normal, log-normal, skewed student-‎\textit{t}‎, Laplace, exponential, Weibull, Gumbel, Birnbaum-Saunders, gamma, Pareto and logistic distributions. Other generalizations are McDonald gamma distribution by
\cite{ma-na-sa-co-12},
McDonald inverted beta distribution by
\cite{co-le-12},
McDonald normal distribution by
\cite{co-ci-re-or-12},
McDonald extended exponential distribution by
\cite{co-ha-or-pa-12},
McDonald half-logistic distribution by
\cite{ol-sa-xa-tr-co-13},
McDonald Dagum by
\cite{ol-ra-13},
McDonald generalized beta-binomial   distribution by
\cite{ma-wi-ya-13},
McDonald log-logistic distribution by
\cite{ta-ma-zu-ha-14},
McDonald arcsine distribution by
\cite{co-le-14},
and McDonald Weibull  distribution by
\cite{co-ha-or-14}.
One of the advantages of the McDonald generated distribution lies in its ability of fitting skewed data such as other well-known distributions in the literature
\cite{mu-na-02,mu-wa-07}.

In this paper, we introduce a new five-parameter model called the ‎\textit{McDonald Gompertz}‎ (McG) distribution that includes as special sub-models some recent distributions in the literature. This distribution offers a more flexible distribution for modeling lifetime data in terms of its hazard rate shapes that are decreasing, increasing, upside-down bathtub and bathtub shaped. Several mathematical properties of this new model in order to attract wider applications in reliability, engineering and in other areas of research are provided.

This paper is organized as follows. In Section \ref{sec.model}, we introduce the McG distribution, density and hazard functions. Some special models of the new distribution are described in this section. In Section \ref{sec.pro}, we present useful expansions and properties of the cumulative distribution function (cdf), probability density function (pdf), ‎\textit{k}‎th moment and moment generating function of the McG distribution. Moreover, order statistics and their moments, entropy and quantile measures are provided in this section. Estimation of the McG parameters by maximum likelihood (ML) method is described in Section \ref{sec.est}. Finally, application of the McG model using two real data sets are considered in Section \ref{sec.exa}.

\section{The McG model}
\label{sec.model}
The generalized beta distribution of the first kind (or beta type I) or McDonald distribution was introduced by
\cite{mcdonald-84}.
The cdf of the McDonald  distribution is given by
\begin{equation*}
F(x)=‎‎I(x^c; a/c, b), \ \ 0<x<1, \ \ a, b, c>0,
\end{equation*}
where $I(y; a, b)=\frac{B_{y}(a, b)}{B(a,b)}
=‎\frac{1}{B(a,b)} ‎\int_{0}^{y} w^{a-1} (1-w)^{b-1} dw$ is the incomplete beta function ratio of type I and $B(a,b)=\int_{0}^{1} w^{a-1} (1-w)^{b-1} dw$ is the beta function.


The cdf of McG model can be defined by
\begin{equation}\label{eq.FMCG}
F(y; a, b, c, \theta, \gamma)=‎‎I([1-\exp (‎-\frac{\theta}{\gamma}‎(e^{\gamma y}-1))]^c; a/c, b),  \ \ y>0,
\end{equation}
where $\theta, \gamma>0 $. The pdf corresponding to \eqref{eq.FMCG} is given by
\begin{eqnarray}\label{eq.fMCG}
f(y; a, b, c, \theta, \gamma)&=&\frac{c \theta e^{\gamma y} }{B(a/c,b)}‎‎‎
\exp (‎-\frac{\theta}{\gamma}‎(e^{\gamma y}-1))[1-\exp (‎-\frac{\theta}{\gamma}‎(e^{\gamma y}-1))]^{a-1} \nonumber\\
&&\times\{1- [1-\exp (‎-\frac{\theta}{\gamma}‎(e^{\gamma y}-1))]^c \}^{b-1}.
\end{eqnarray}
Here after, we denote a random variable $Y$ with pdf in \eqref{eq.fMCG} by $McG (a, b, c, \theta, \gamma)‎$. Indeed, the McG distribution
belongs to McDonald-generalized  class of distributions with cdf and pdf as
\begin{equation*}
F(y)=‎‎I(G^c(y); a/c, b)= ‎\frac{1}{B(a/c, b)} ‎\int_{0}^{G^c(y)} w^{a/c-1} (1-w)^{b-1} dw,
\end{equation*}
and
\begin{equation*}
f(y)=\frac{c}{B(a/c, b)} g(y) G^{a-1}(y) (1-G^c(y))^{b-1},
\end{equation*}
respectively. The cdf in \eqref{eq.FMCG} can be expressed in terms of the hypergeometric function as
\begin{equation*}
F(y; a, b, c, \theta, \gamma)=‎\frac{c G^a(y)}{a B(a/c, b)}‎ _{2}F_{1} (a/c, 1-b; a/c+1, G^c(y) ),
\end{equation*}
where
$_{2}F_{1} (a, b; c, x)=\sum_{n=0}^{\infty} \frac{(a)_{n} (b)_{n}}{(c)_{n}} ‎\frac{x^n}{n!}‎$
is ascending factorial, and $G(y)= 1-\exp (‎-\frac{\theta}{\gamma}‎(e^{\gamma y}-1))$ is the cdf of G distribution.

\begin{theorem}
Let $f(y; a, b, c, \theta, \gamma)$ be the pdf of McG distribution given by \eqref{eq.fMCG}. The limiting behavior of $f(y; a, b, c, \theta, \gamma)$ for different values of its parameters is given below:

\noindent $i$. If $a=1$, then ${\mathop{\lim }_{y \rightarrow 0^+} f (y; a, b, c, \theta, \gamma)}=\frac{  \theta c}{B(1/c,b)}.$

\noindent $ii$. If $a>1$, then ${\mathop{\lim }_{y \rightarrow 0^+} f (y; a, b, c, \theta, \gamma)}=0.$

\noindent $iii$. If $a<1$, then ${\mathop{\lim }_{y \rightarrow 0^+} f (y; a, b, c, \theta, \gamma)}=\infty.$

\noindent $iv$. ${\mathop{\lim }_{y \rightarrow + \infty} f (y; a, b, c, \theta, \gamma)}=0.$
\end{theorem}

\begin{proof}
The parts (i)-(iii) are obviously proved. For part (iv), we have
\begin{eqnarray*}
&&0\leq [1-[1-exp\{-\frac{\theta}{\gamma}(e^{\gamma y}-1)\}]^{c}]^{b-1}<1  \Longrightarrow \\
&&0<f(y; a, b, c, \theta, \gamma)<{c \theta e^{\gamma y} \exp(-\frac{\theta}{\gamma}(e^{\gamma y}-1))[1-\exp(-\frac{\theta}{\gamma}(e^{\gamma y}-1))]^{a-1}}/{B(a/c,b)}.
\end{eqnarray*}
It can be easily shown that
$${\mathop{\lim }_{y \rightarrow \infty} c \theta e^{\gamma y} \exp(-\frac{\theta}{\gamma}(e^{\gamma y}-1))[1-\exp(-\frac{\theta}{\gamma}(e^{\gamma y}-1))]^{a-1}}=0,$$
and the proof is completed.
\end{proof}

\begin{figure}[ht]
\centering
\includegraphics[scale=0.33]{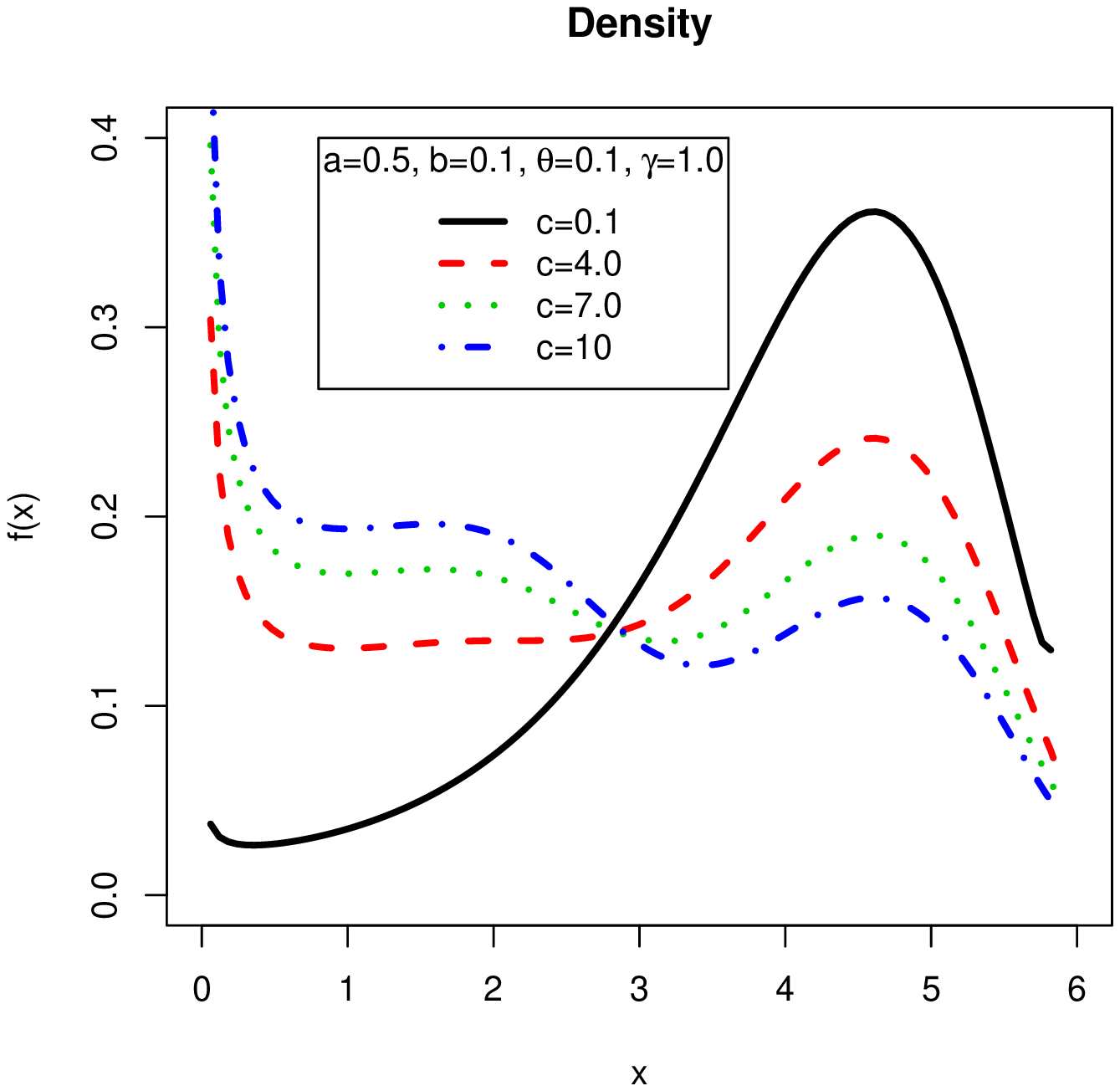}
\includegraphics[scale=0.33]{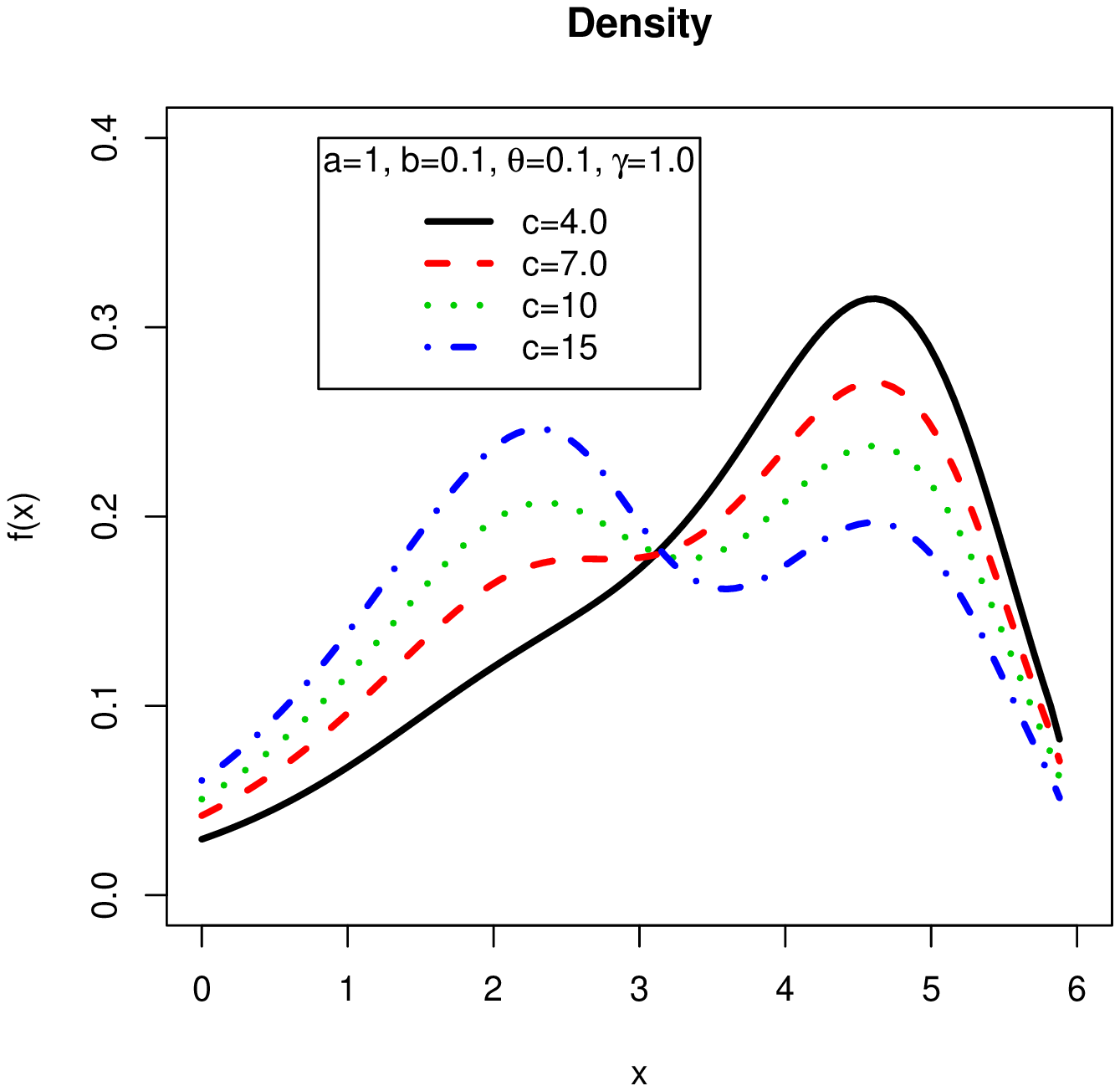}
\includegraphics[scale=0.33]{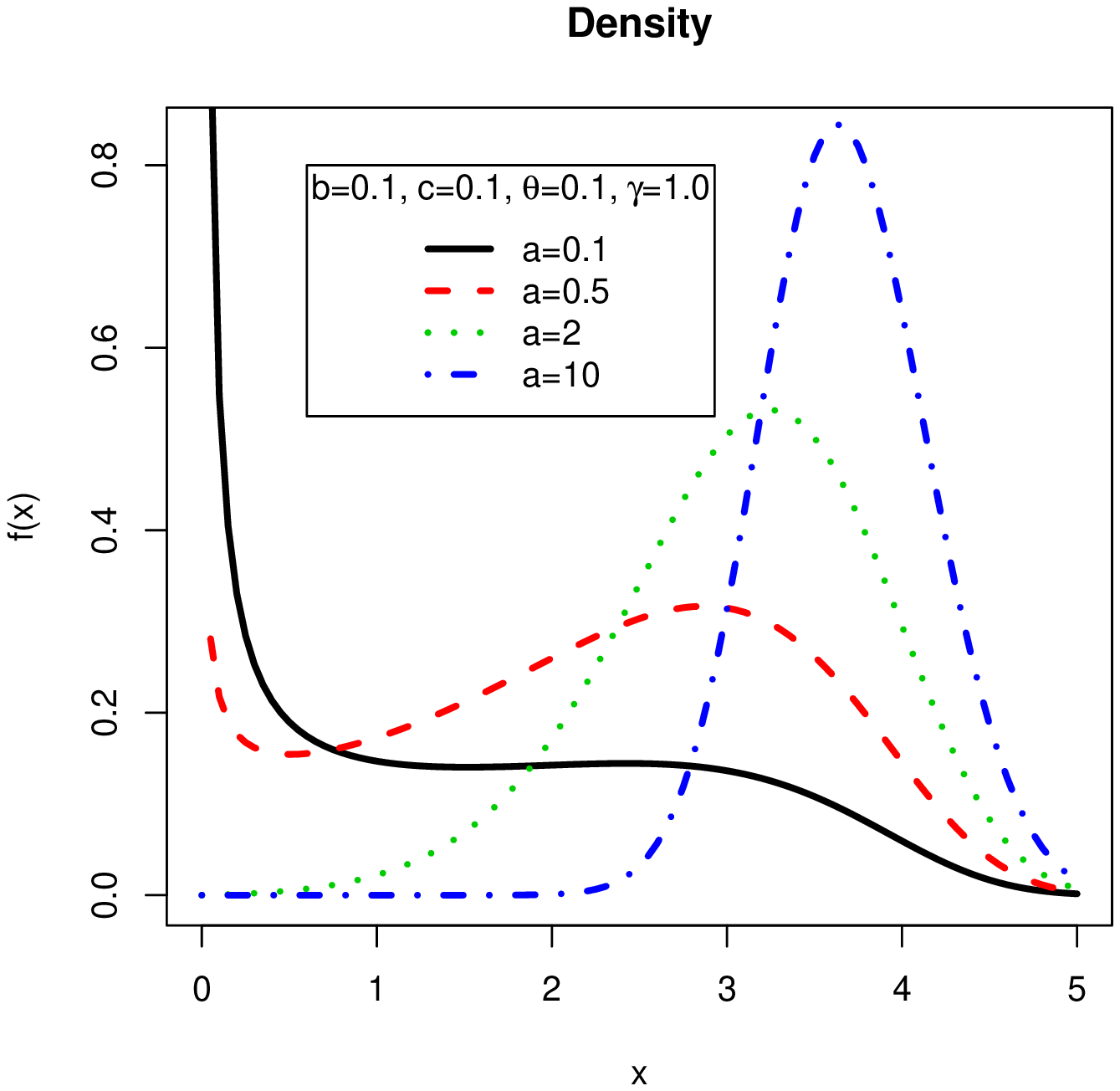}
\includegraphics[scale=0.33]{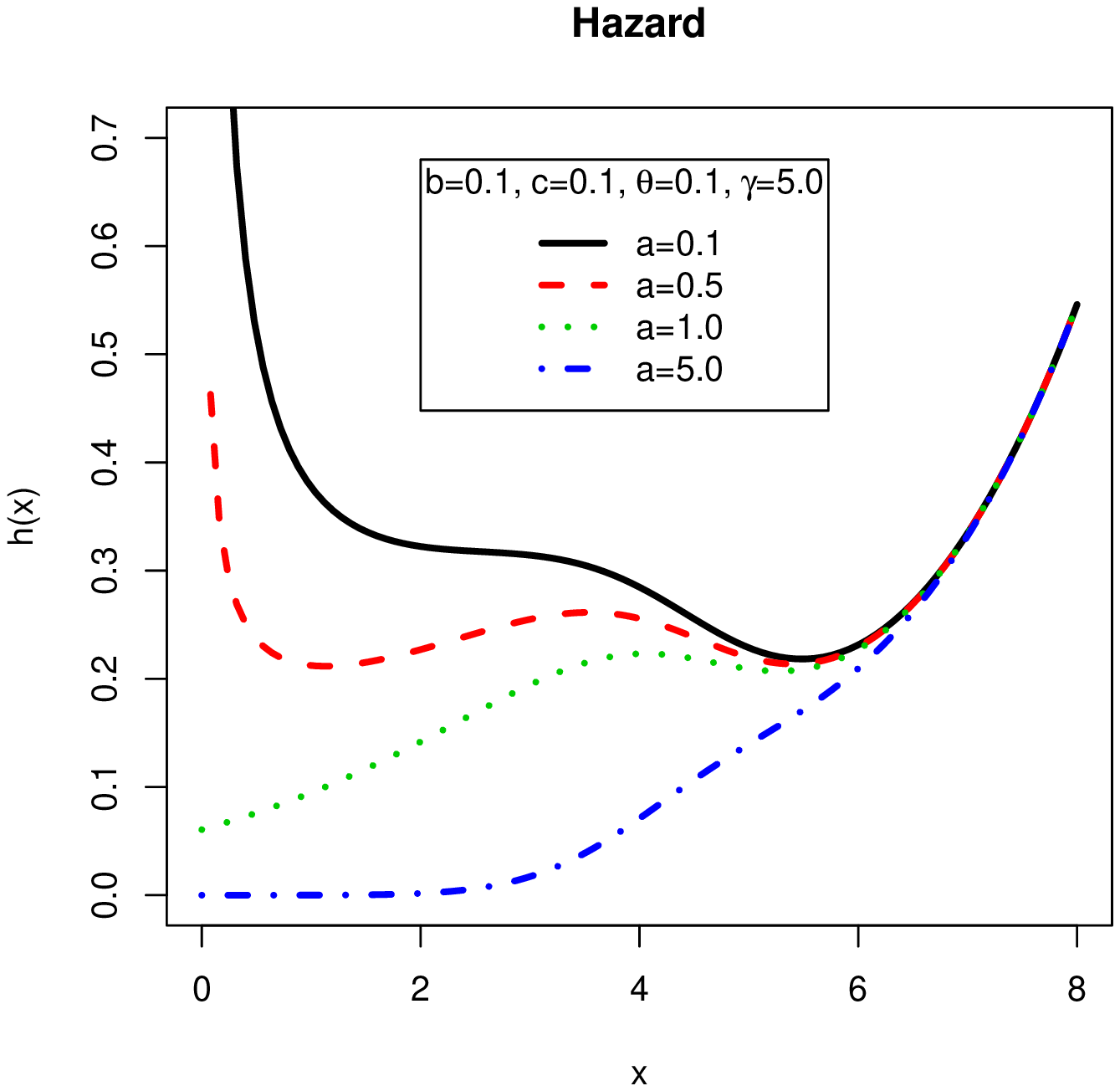}
\includegraphics[scale=0.33]{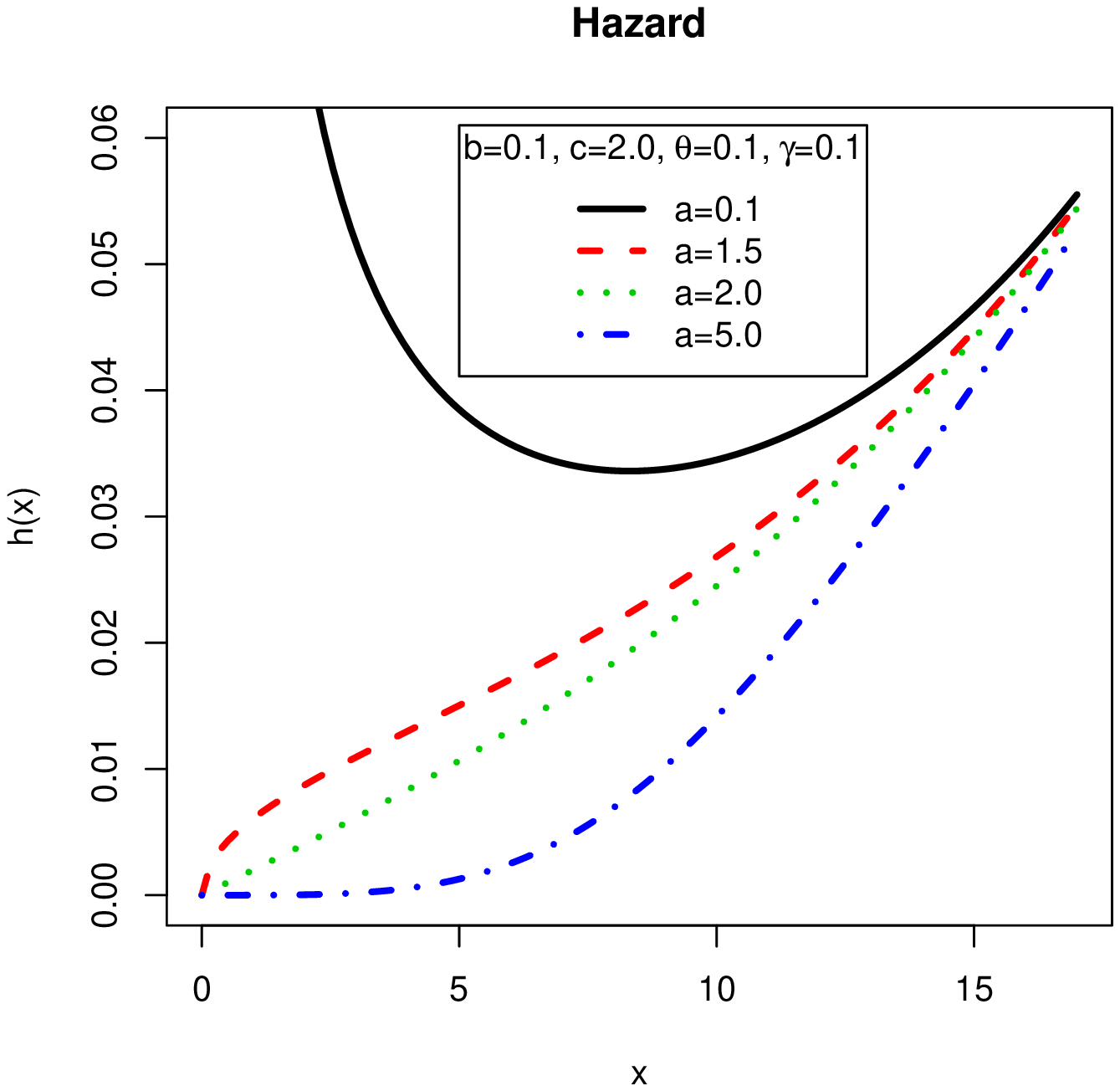}
\includegraphics[scale=0.33]{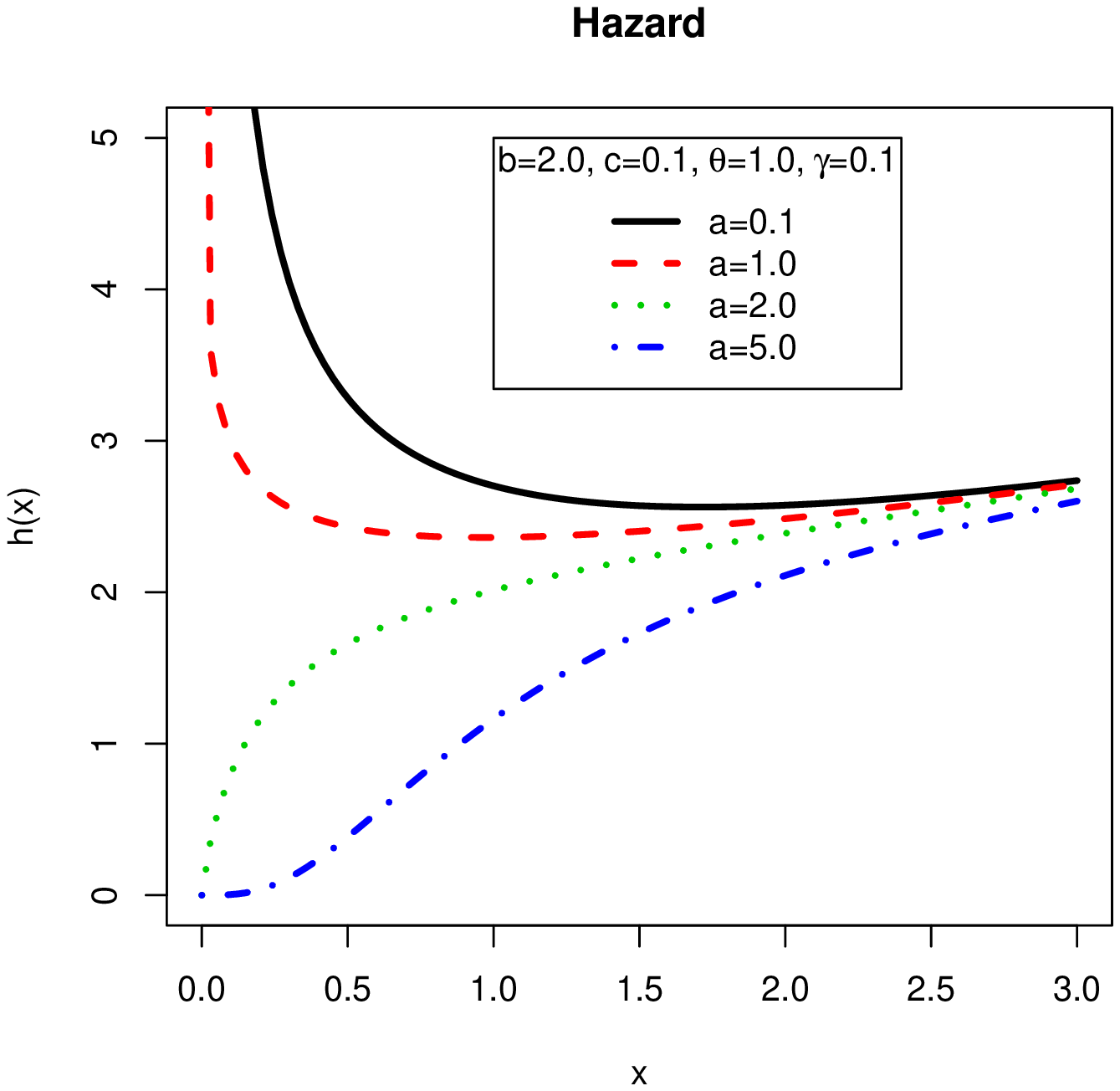}
\vspace{-0.8cm}
\caption[]{pdf and hrf of McG model for some values of parameters.}\label{plot.denh}
\end{figure}

From equations \eqref{eq.FMCG} and \eqref{eq.fMCG}, it is easy to verify that the hazard rate function (hrf) of the McG distribution is given by
\begin{eqnarray*}
h(y; a, b, c, \theta, \gamma) &=& \frac{c \theta e^{\gamma y} \exp(-\frac{\theta}{\gamma}(e^{\gamma y}-1)) } {B(a/c,b)-B_{[1-\exp(-\frac{\theta}{\gamma}(e^{\gamma y}-1))]^c}(a/c,b)} \\
&‎&\times [1-\exp(-\frac{\theta}{\gamma}(e^{\gamma y}-1))]^{a-1} [1-(1-\exp(-\frac{\theta}{\gamma}(e^{\gamma y}-1)))^{c}]^{b-1}, \  \ y>0,
\end{eqnarray*}
and the corresponding reversed hazard rate function reduces to
\begin{eqnarray*}
r(y; a, b, c, \theta, \gamma) &=& \frac{c \theta e^{\gamma y} \exp(-\frac{\theta}{\gamma}(e^{\gamma y}-1)) } {B_{[1-\exp(-\frac{\theta}{\gamma}(e^{\gamma y}-1))]^c} (a/c,b)} [1-\exp(-\frac{\theta}{\gamma}(e^{\gamma y}-1))]^{a-1}\\
&&\times ‎ [1-(1-\exp(-\frac{\theta}{\gamma}(e^{\gamma y}-1)))^{c}]^{b-1},  \ \ \  \ y>0.
\end{eqnarray*}

Figure \ref{plot.denh}  illustrates some of the possible shapes of density and hazard functions for selected values of parameters. For instance, these plots show the hazard function of the new model is much more flexible than the beta Gompertz (BG) and G distributions. The hazard rate function can be bathtub shaped, monotonically increasing or decreasing and upside-down bathtub shaped depending on the parameter values.

The McG distribution contains as sub-models the Kumaraswamy Gompertz (KumG), the BG
\cite{ja-ta-al-14},
and the beta generalized exponential (BGE) or McE
\cite{ba-sa-co-10}
distributions for $c=1$, $a=c$ and $\gamma‎\rightarrow‎ 0$, respectively. It also contains the beta exponential (BE)
\cite{na-ko-06},
the generalized Gompertz (GG)
\cite{el-al-al-13},
and the Kumaraswamy exponential (KumE)
\cite{na-co-or-12}
distributions. The GE
\cite{gu-ku-99}, the G and E distributions are also sub-models. The classes of distributions that are included as special sub-models of the McG distribution are displayed in Figure \ref{plot.relation}.

\begin{figure}[ht]
\centering
\includegraphics[scale=1.1]{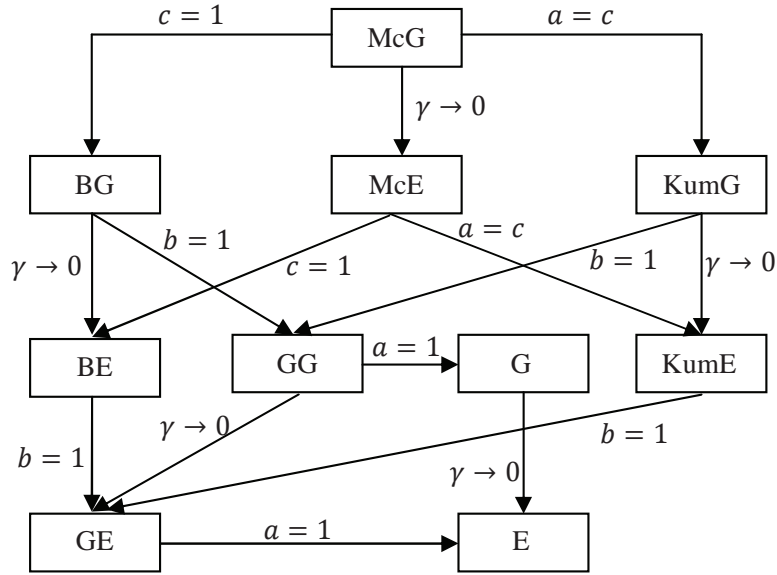}
\caption[]{Relationships of the McG sub-models.}\label{plot.relation}
\end{figure}

If the random variable $Y$ has the McG distribution, then it has the following properties:\\
1. The random variable $V=[1-\exp(-\frac{\theta}{\gamma}(e^{\gamma Y}-1))]^{c}$ satisfies the beta distribution with parameters $a/c$ and $b$. Therefore, the random variable $T=\frac{\theta}{\gamma}(e^{\gamma Y}-1)$ has the BGE (or McG) distribution
\cite{ba-sa-co-10}.
Furthermore, the random variable $Y=G^{-1}(V^{‎1/c})=\frac{1}{\gamma}‎ \log [1-\frac{\gamma}{\theta} \log(1-V^{1/c})‎]$ follows McG distribution. This result helps us in simulating data from McG distribution.
The plots comparing the exact McG density function and the histogram from a simulated data set with size 100 for some parameter values are given in Figure \ref{plot.sim} (left). Also, the plots of empirical distribution function and exact distribution function are given in Figure \ref{plot.sim} (right). These plots indicate that the simulated values are consistent with the McG distribution.\\
2. If $a=i$ and $b=n-i+1$, where $i$ and $n$ are positive integer values, then the $F(y; a, b, c, \theta, \gamma)$  is the cdf of the $i$th order statistic of GG distribution.

\begin{figure}[ht]
\centering
\includegraphics[scale=0.5]{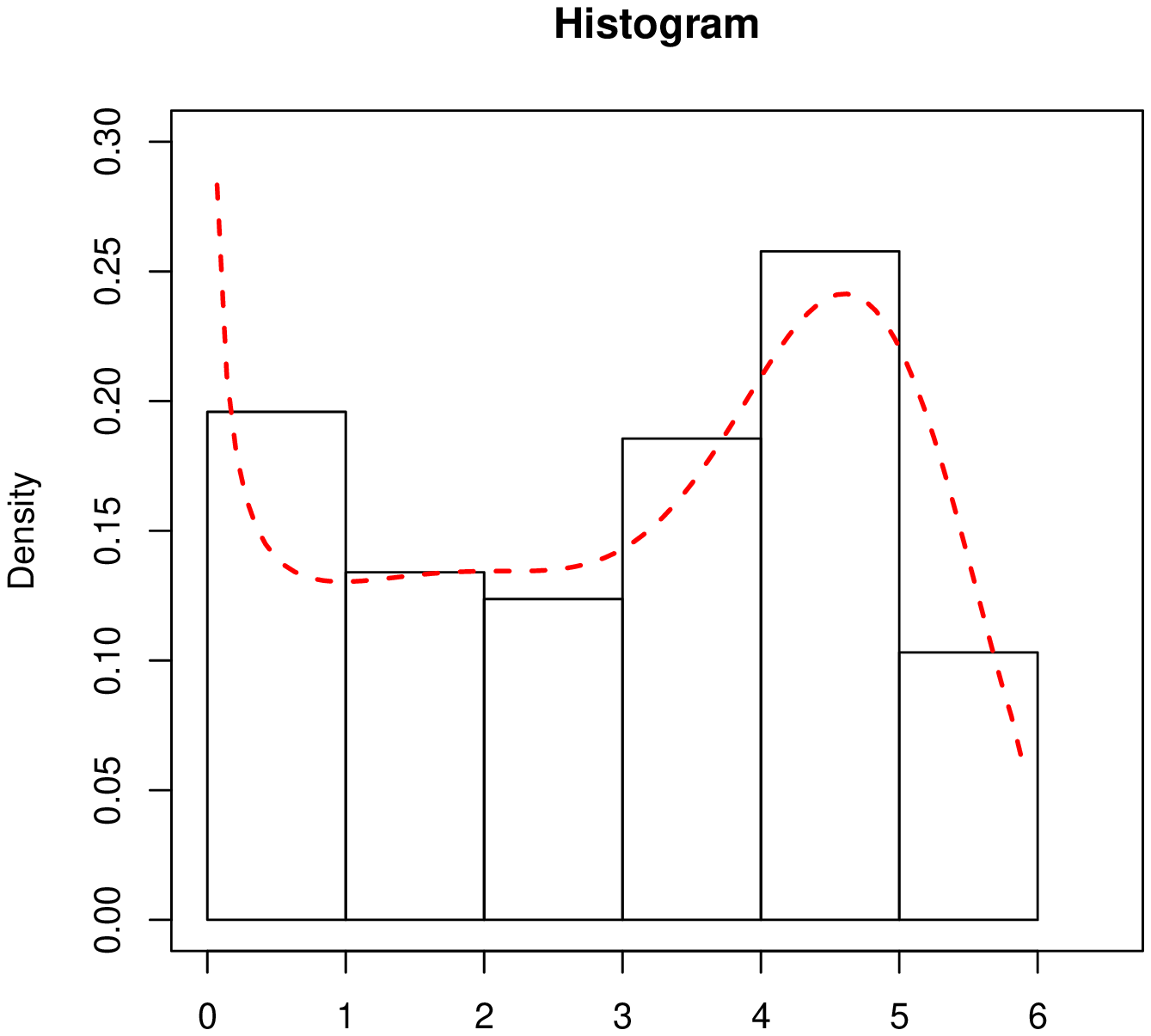}
\includegraphics[scale=0.5]{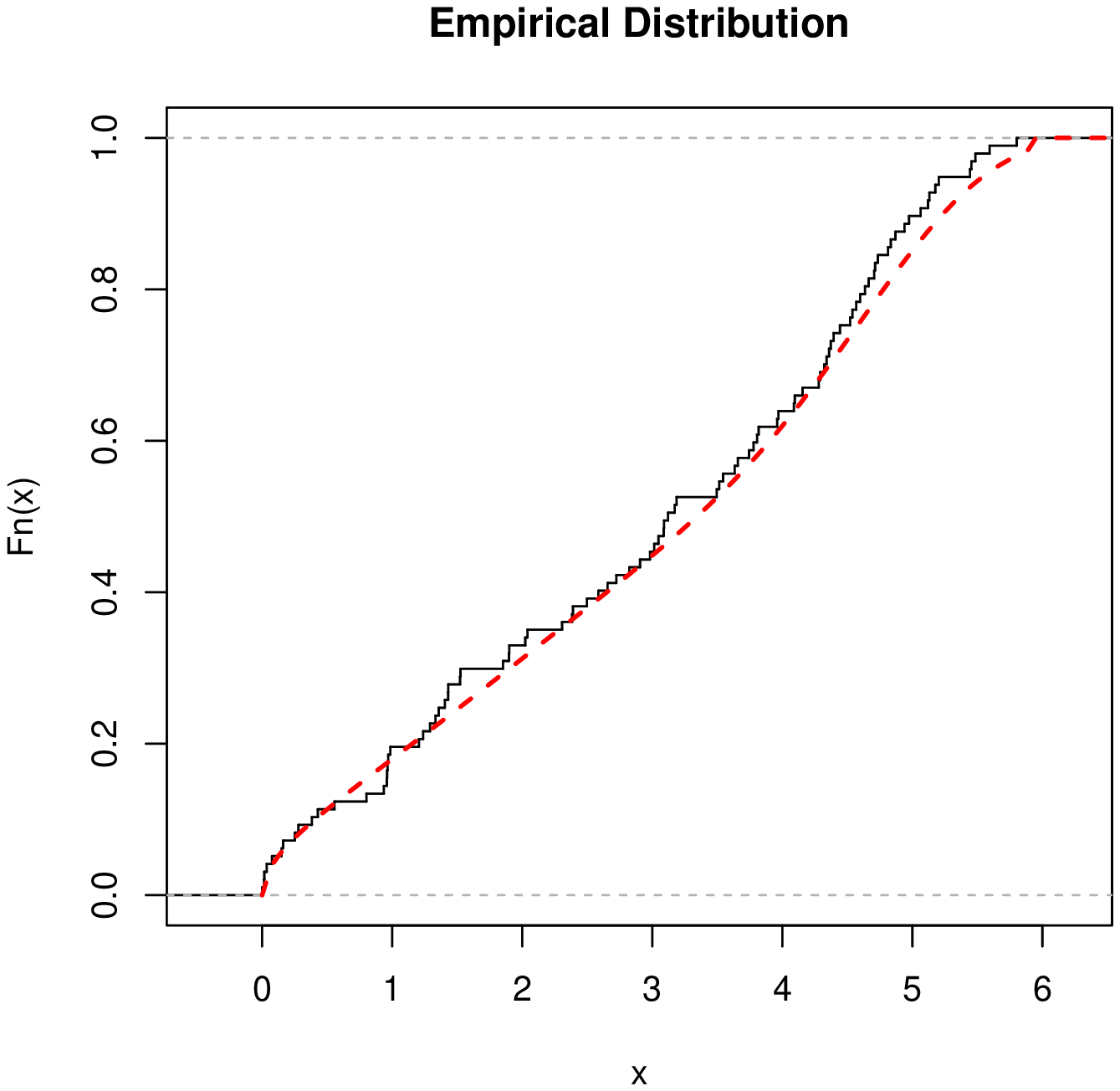}

\vspace{-1cm}
\caption[]{The histogram of a generated data set with size 100 and the exact McG density (left) and the empirical distribution function and exact distribution function (right).}\label{plot.sim}
\end{figure}

\section{General properties}
\label{sec.pro}
In this section, some properties of McG distribution are considered.

\subsection{A useful expansion}

We derive some expansions for the cdf, $k$th moment and moment generating function of the McG distribution. The binomial series expansion is defined by
\begin{equation}\label{eq.binom}
(1-z)^m=\sum_{j=0}^{\infty} (-1)^j \binom {j} {m} z^j= \sum_{j=0}^{\infty} (-1)^j ‎\frac{\Gamma(m+1)}{\Gamma(m-j+1)}‎ ‎\frac{z^j}{j!}‎,
\end{equation}
where $|z|<1$ and $m$ is a positive real non-integer.

The following proposition reveals that the McG distribution can be expressed as a mixture of distribution function of GG distribution, whereas Proposition \ref{prop.mixgg} provides a useful expansion for the pdf in \eqref{eq.fMCG}.

\begin{prop}
The cdf in \eqref{eq.FMCG} is a mixture of distribution function of GG distribution on the form
\begin{equation}\label{eq.Fmix}
F(y; a, b, c, \theta, \gamma)=\sum_{j=0}^{\infty} p_{j} [G(y)]^{a+jc}=\sum_{j=0}^{\infty} p_{j} G_{j}(y),
\end{equation}
where $p_{j}=\frac{(-1)^{j}\Gamma(b)}{ B(a/c,b) \Gamma(b-j) j! (a/c+j)}$ and $G_{j}(y)=(G(y))^{a+jc}$ is the distribution function of a random variable which has a GG distribution with parameters $\theta$, $\gamma$ and $a+jc$.
\end{prop}
using binomial expansion \eqref{eq.binom} to the term $[G(y)]^{a+jc}$ in \eqref{eq.Fmix}, we have
\begin{eqnarray*}
[G(y)]^{a +jc} &=& \sum_{k=0}^{\infty} (-1)^k \binom {a+jc} {k} (1-G(y))^k \\
&=‎& \sum_{r=0}^{\infty} \sum_{k=r}^{\infty} (-1)^k \binom {a+jc} {k} \binom {k} {r} [G(y)]^r.
\end{eqnarray*}
Now, \eqref{eq.Fmix} becomes
\begin{eqnarray*}
F(y; a, b, c, \theta, \gamma) = \sum_{j=0}^{\infty} \sum_{r=0}^{\infty} \sum_{k=r}^{\infty} p_j (-1)^{k+r} \binom {a+jc} {k} \binom {k} {r} [G(y)]^r
=‎\sum_{r=0}^{\infty} b_r [G(y)]^r,
\end{eqnarray*}
where $b_r= \sum_{j=0}^{\infty} \sum_{k=r}^{\infty} p_j (-1)^{k+r} \binom {a+jc} {k} \binom {k} {r}$.

\begin{prop}\label{prop.mixgg}
The pdf of McG can be expressed as an infinite mixture of GG densities with parameters $\theta$, $\gamma$ and $(a+jc)$ given by
\begin{eqnarray*}
f(y; a, b, c, \theta, \gamma) = \sum_{j=0}^{\infty} p_j (a+jc) g(y) [G(y)]^{a+jc-1}
=‎ \sum_{r=0}^{\infty} p_j g_j (y),
\end{eqnarray*}
where $g_j (y)= (a+jc) g(y) [G(y)]^{a+jc-1}$. We can write the pdf of McG as
\begin{equation*}
f(y; a, b, c, \theta, \gamma)=g(y) \sum_{r=0}^{\infty} c_r [G(y)]^r,
\end{equation*}
where $c_r= \sum_{j=0}^{\infty} \sum_{k=r}^{\infty} ‎(-1)^{j+k+r} \frac{c \Gamma (b)}{B(a/c,b) j! \Gamma(b-j)}  \binom {a+jc-1} {k} \binom {k} {r}$.
\end{prop}

\subsection{Moments and generating function}

In this section, we deal with the basic statistical properties of McG distribution such as the $k$-th moment and generating function in the following propositions.

\begin{prop}
The $k$-th moment of McG distribution can be expressed as a infinite mixture of the $k$-th moment of GG distributions as follows:
\begin{eqnarray*}
E(Y^{k}) = \int_{0}^{\infty} y^{k} \sum_{j=0}^{\infty} p_{j} (a+jc) g(y) [G(y)]^{a+jc-1}
= \sum_{j=0}^{\infty} p_{j} E(Y_{j}^{k}),
\end{eqnarray*}
where
\begin{equation*}
E(Y_{j}^{k}) = u_{jk} \sum_{i=0}^{\infty} \sum_{r=0}^{\infty} \binom {a+jc-1} {i} ‎\frac{(-1)^{i+r}}{\Gamma (r+1)}
e^{‎\frac{\theta}{\gamma}‎ (i+1)} ‎[\frac{\theta}{\gamma}‎ (i+1)]^r [‎\frac{-1}{\gamma (r+1)}‎]^{k+1},
\end{equation*}
and $u_{jk}= \theta (a+jc) \Gamma (k+1)$.
\end{prop}

\begin{prop}
An explicit expression for the moment generating function of McG distribution follows from Proposition \ref{prop.mixgg},
\begin{eqnarray*}
M_{Y} (t) =\int_{0}^{\infty} e^{tx} \sum_{j=0}^{\infty} p_j (a+jc) g(y) [G(y)]^{a+jc-1}
=‎ \sum_{j=0}^{\infty} p_j  M_{Y_{j}} (t),
\end{eqnarray*}
where
\begin{equation*}
M_{Y_j} (t)= ‎\frac{(a+jc) \theta}{\gamma}‎ \sum_{i=0}^{\infty} \sum_{k=0}^{\infty} (-1)^i \binom {a+jc-1} {i} \binom {t/{\gamma}} {k} ‎\frac{\Gamma(k+1)}{[\frac{(a+jc) \theta}{\gamma}]^{k+1}}‎ .
\end{equation*}
\end{prop}

\subsection{Order statistics}

Order statistics make their appearance in many areas of statistical theory and practice. Let the random variable $Y_{i:n}$ be the $i$th order statistic ($Y_{1:n}‎\leq Y_{2:n}‎\leq \cdots ‎\leq Y_{n:n}‎ ‎‎$) in a sample of size $n$ from the McG distribution. The pdf and cdf of $Y_{i:n}$ for $i=1,2, \ldots, n$ are given by
\begin{eqnarray}\label{eq.fin}
f_{i:n}(y) &=& \frac{1}{B(i,n-i+1)} f(y) [F(y)]^{i-1} [1-F(y)]^{n-i} \nonumber\\
&=& \frac{1}{B(i,n-i+1)}\sum_{k=0}^{n-i} \dbinom{n-i}{k} (-1)^{k} f(y) [F(y)]^{k+i-1},
\end{eqnarray}
and
\begin{eqnarray}\label{eq.Fin}
F_{i:n}(y) = \int_{0}^{y} f_{i:n}(t)dt
= \frac{1}{B(i,n-i+1)} \sum_{k=0}^{n-i} \frac{(-1)^{k} }{k+i}\dbinom{n-i}{k} [F(y)]^{k+i},
\end{eqnarray}
respectively, where  $F(y)= \sum_{r=0}^{\infty} b_{r} G(y)$. We use throughout an equation by
(\cite{gr-ry-07}, page 17)
for a power series raised to a positive integer $m$ given by
\begin{equation}\label{eq.pow}
\left(\sum_{r=0}^\infty b_r\, u^r\right)^m=\sum_{r=0}^\infty c_{m,r}\,u^r,
\end{equation}
where the coefficients $c_{m,r}$ (for $r=1,2,\ldots$) are easily determined from the recurrence equation
\begin{equation*}
c_{m,r}=(r\,b_0)^{-1} \sum_{k=1}^{r} \,[k\,(m+1)-r+k]\,b_k\,c_{m,r-k},
\end{equation*}
where $c_{m,0}=b_0^m$. Hence, the coefficients $c_{m,r}$ can be calculated from $c_{m,0},\ldots,c_{m,r-1}$ and therefore, from the quantities $b_0,\ldots,b_{r}$.  Using \eqref{eq.pow}, the equations \eqref{eq.fin} and \eqref{eq.Fin} can be written as
\begin{eqnarray*}
&&f_{i:n}(y)=\frac{1}{B(i,n-i+1)} \sum_{k=0}^{n-i} \sum_{r=1}^{\infty} \frac{r}{k+i}  (-1)^{k} \dbinom{n-i}{k} c_{i+k,r} g(y) [G(y)]^{r-1},\\
&&F_{i:n}(y)=\frac{1}{B(i,n-i+1)} \sum_{k=0}^{n-i} \sum_{r=0}^{\infty} \frac{1}{k+i} (-1)^{k} \dbinom{n-i}{k} c_{i+k,r} [G(y)]^{r}.
\end{eqnarray*}
An explicit expression for the $s$th moments of $Y_{i:n}$ can be obtained as
\begin{eqnarray}
E[Y_{i:n}^{s}] &=& \frac{1}{B(i,n-i+1)} \sum_{k=0}^{n-i} \sum_{r=1}^{\infty} \frac{r}{k+i} (-1)^{k} \dbinom{n-i}{k} c_{i+k,r} \int_{0}^{+\infty} t^{s} g(t) [G(t)]^{r-1} dt \nonumber\\
&=&\frac{\theta \Gamma(s+1)}{B(i,n-i+1)} \sum_{k=0}^{n-i} \sum_{r=1}^{\infty} \frac{r}{k+i} (-1)^{k} \dbinom{n-i}{k} c_{i+k,r}\nonumber\\
&&\times
\sum_{i_{1}=0}^{\infty} \sum_{i_{2}=0}^{\infty} \dbinom {r\lambda-1} {i_{1}} \frac{(-1)^{i_{1}+i_{2}}}{\Gamma(i_{2}+1)} e^{\frac{\theta} {\gamma}(i_{1}+1)} [\frac{\theta(i_{1}+1)} {\gamma}]^{i_{2}}[\frac{-1}{\gamma(i_{2}+1)}]^{s+1}.
\end{eqnarray}

\subsection{Quantile measures}

In this section, we consider the effect of each shape parameters $a$, $b$ and $c$ on the skewness and kurtosis of the McG distribution. To illustrate this effect, we use measures based on quantiles.
The quantile function of the $McG (a,b,c,\gamma, \theta)$ distribution say $Q(t)$ can be obtained as
$$Q(t)=\frac{1}{\gamma} \log (1-\frac{\gamma}{\theta} \log (1-Q_{a/c,b}^{\frac{1}{c}}(t))), \  \ 0<t<1,$$
where $Q_{a/c,b}(t)=I_{t}^{-1}(a/c,b)$ denotes the $t$th quantile of beta distribution with parameters $a/c$ and $b$. The Bowley skewness
(see \cite{ke-ke-62})
based on quantiles can be calculated by
$$ {\mathcal{B}}=\frac{Q(\frac{3}{4})-2Q(\frac{1}{2})+Q(\frac{1}{4})}{Q(\frac{3}{4})-Q(\frac{1}{4})}, $$
and the Moors kurtosis
(see \cite{moors-88})
is defined as
$$ {\mathcal{M}}=\frac{Q(\frac{7}{8})-Q(\frac{5}{8})+Q(\frac{3}{8})-Q(\frac{1}{8})}{Q(\frac{6}{8})-Q(\frac{2}{8})}, $$
where $Q(.)$ denotes the quantile function. These measures are less sensitive to outliers and they exist even for distributions without moments.
For the standard normal  and  the classical standard $t$ distributions with 10 degrees of freedom, the Bowley measure is 0.
The Moors measure  for these 
 distributions
is  1.2331 and 1.27705, respectively.


In Figure \ref{plot.skur1}, we plot Bowley measure (holding $b = 0.5$, $\gamma = 1$ and $\theta=0.1$ fixed) as a function of $c$ for fixed values of $a$ (left) and Bowley measure (holding $a = 0.5$, $\gamma = 1$ and $\theta=0.1$ fixed) as a function of $c$ for some values of $b$ (right). In Figure \ref{plot.skur2}, we plot Moors measure (for $b = 0.5$, $\gamma = 1$ and $\theta=0.1$) as a function of $c$ for selected values of $a$ (left) and Moors measure (for $a = 0.5$, $\gamma = 1$ and $\theta=0.1$) as a function of $c$ for some values of $b$ (right). These plots indicate that these measures can be sensitive to the three shape parameters $a$, $b$ and $c$.

\begin{figure}[ht]
\centering
\includegraphics[scale=0.5]{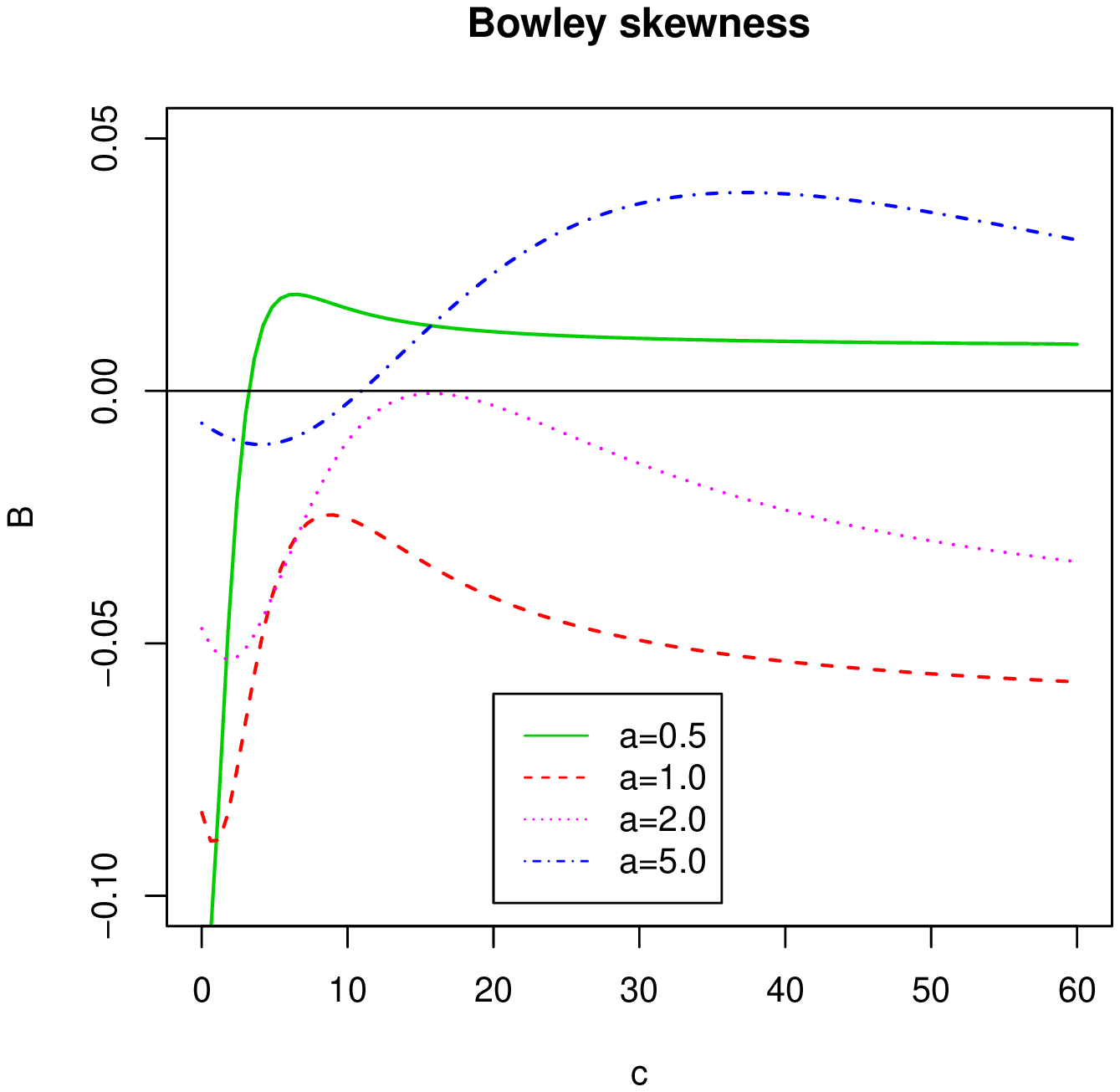}
\includegraphics[scale=0.5]{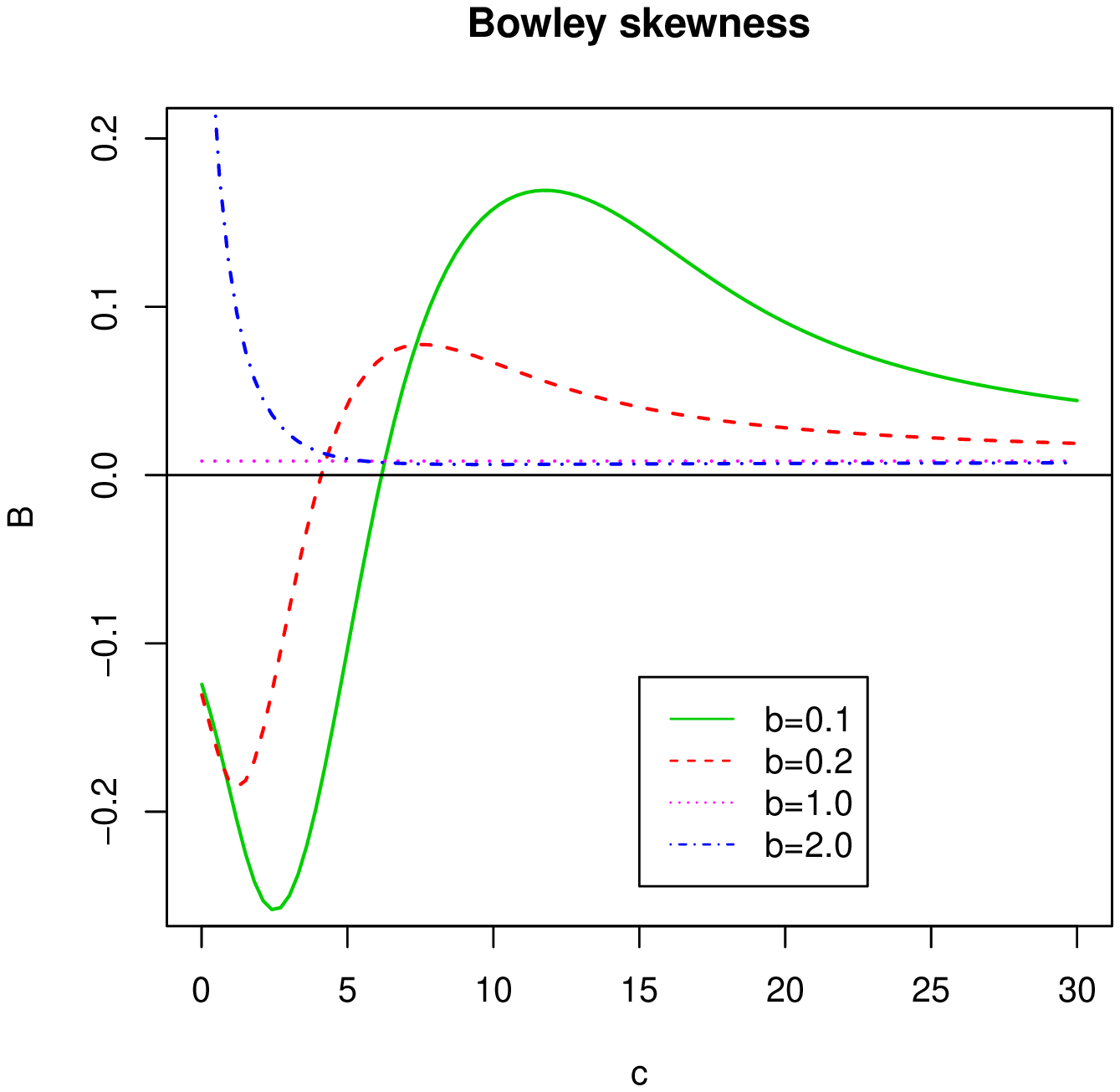}

\vspace{-1cm}
\caption{The Bowley’s skewness of the McG distribution as a function of $c$.}\label{plot.skur1}
\end{figure}

\begin{figure}[ht]
\centering
\includegraphics[scale=0.5]{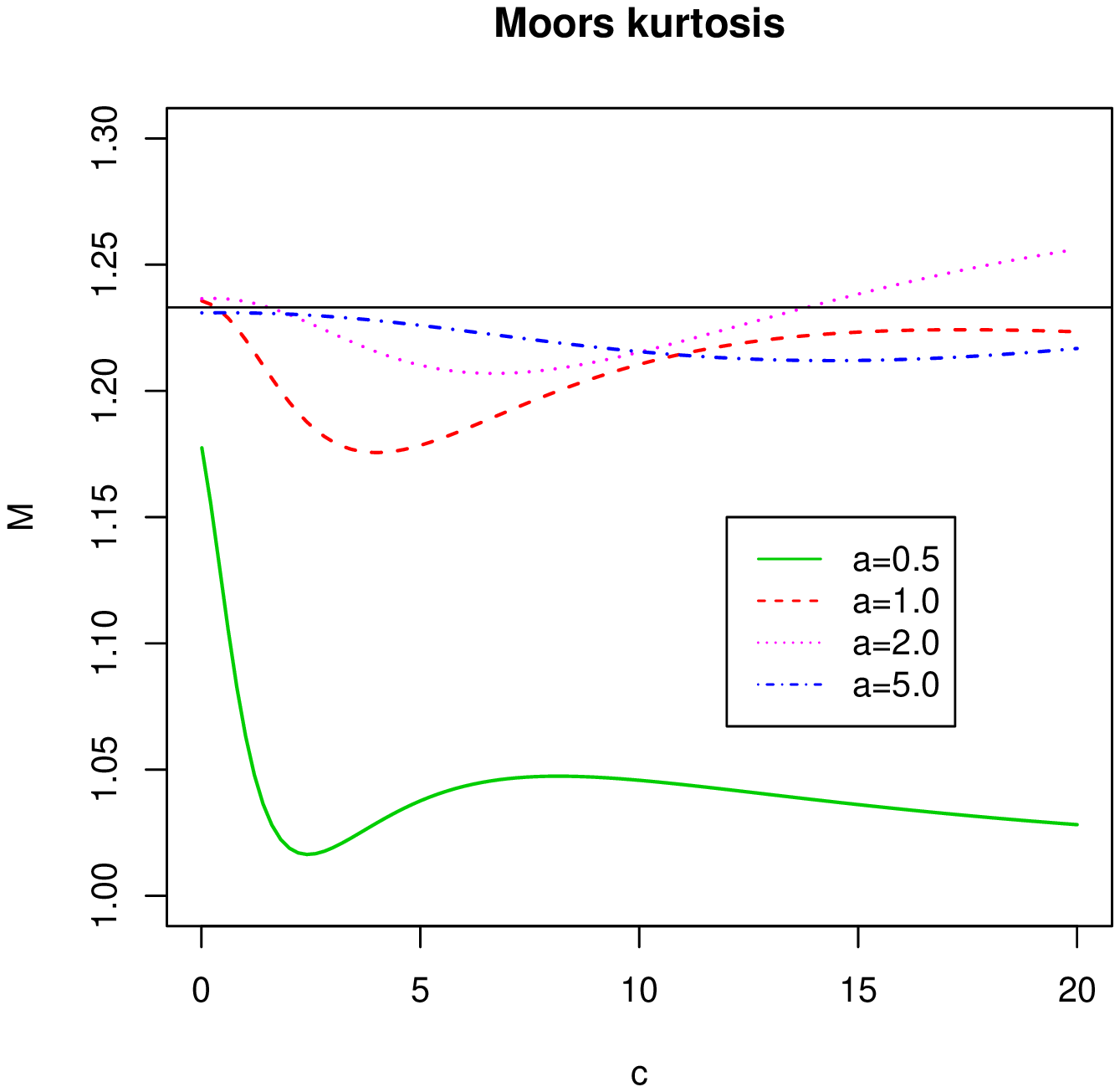}
\includegraphics[scale=0.5]{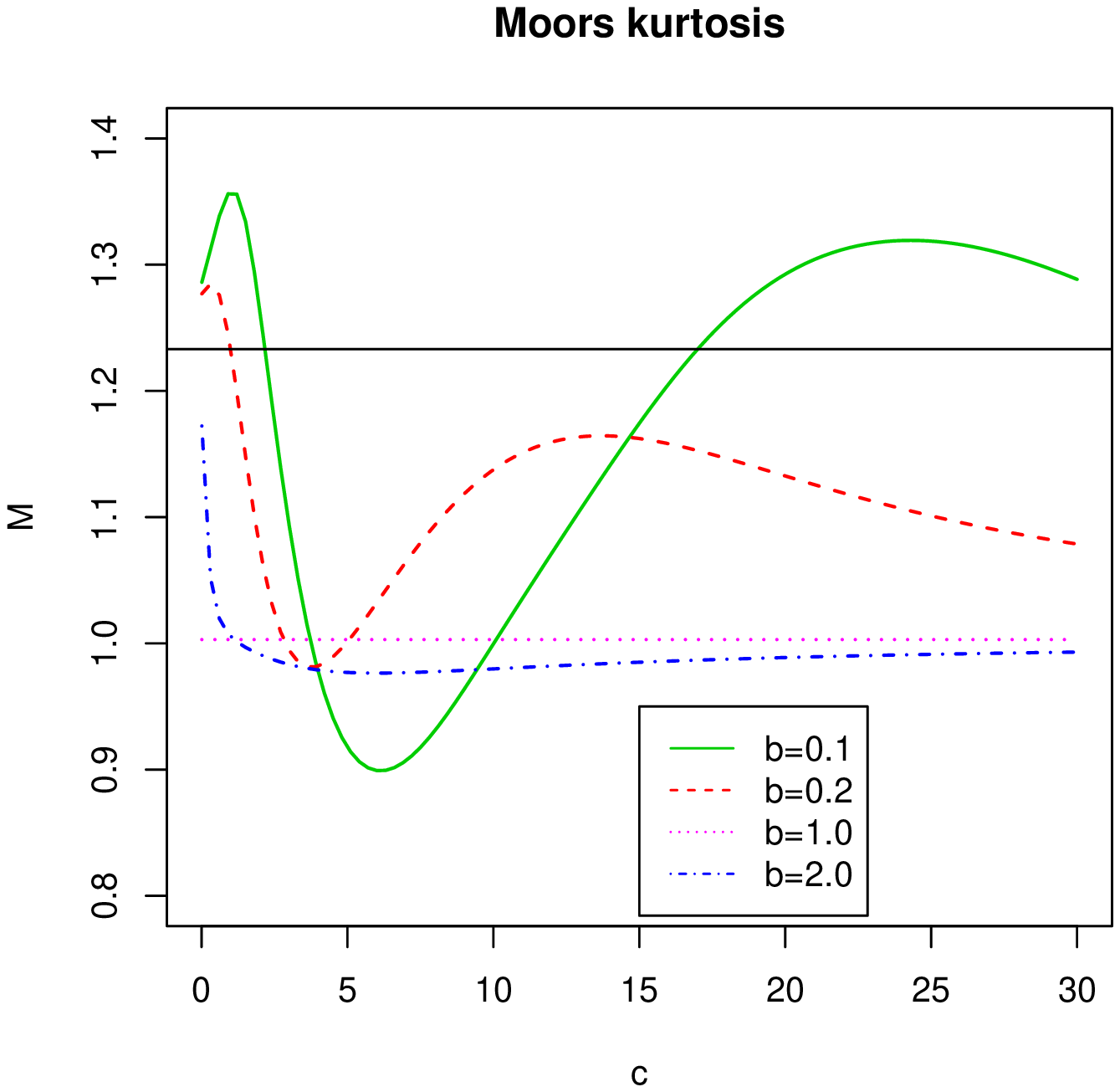}

\vspace{-1cm}
\caption{The Moors’ skewness of the McG distribution as a function of $c$.}\label{plot.skur2}
\end{figure}

\subsection{Entropy}

The entropy of random variable is defined in terms of its probability distribution and can be shown to be a good measure of randomness or uncertainty. The Shannon's entropy of a continuous random variable $Y$ with pdf $f(y)$ is defined by
\cite{shan-48}
as
\begin{equation*}
H_{Sh} (f)= - E_f [\log f(Y)]= - \int_{0}^{\infty} f(y) \log f(y) dy.
\end{equation*}
Hence, the Shannon entropy for McG distribution can be expressed in the form
\begin{equation}\label{eq.HSH}
H_{Sh} (f)= \log (‎\frac{B(a/c,b)}{c \theta}‎) - \theta / ‎\gamma‎ - ‎\gamma‎ E(Y)+ \theta / ‎\gamma M_Y (\gamma) + (a-1) ‎\zeta‎ (a,b) + (b-1) ‎\zeta‎ (b,a),
\end{equation}
where $‎\zeta‎ (r,s)=\psi (r+s) - \psi (r)‎$ and $\psi (.)‎$ represents the digamma function. The last two terms in \eqref{eq.HSH} follows immediately from the first two conditions in Lemma 1 of
\cite{zo-ba-09}.
The R\'{e}nyi entropy is defined by
\begin{equation*}
H_{‎\rho‎} (f)= ‎\frac{1}{1-‎\rho‎}‎  \log ‎( ‎ ‎‎\int_{-\infty}^{\infty} [f(y)]^{‎\rho} dy ),
\end{equation*}
where $\rho >0$ and $\rho ‎\neq 1‎$. The Shannon entropy is derived from $\lim_{\rho ‎\rightarrow 1‎} H_{‎\rho‎} (f)$.
An explicit expression of R\'{e}nyi entropy for McG distribution is obtained as
\begin{eqnarray*}
H_{‎\rho‎} (f) &=& - \log (\theta) + ‎\frac{\rho}{1-\rho} \log (‎\frac{c}{B(a/c,b)}‎) + ‎\frac{1}{1-\rho} \log (B(a \rho- \rho + cj +1, \rho))‎\\
&‎&+ ‎\frac{1}{1-\rho} \log (\sum_{j=0}^{\infty} (-1)^j \binom {b \rho- \rho} {j}) E \{(1-\gamma / \theta \log (1-U))^{\rho -1} \},
\end{eqnarray*}
where $U$ has a beta distribution with parameters $a \rho - \rho + cj -1$ and $\rho$.

\section{Estimation}
\label{sec.est}
Let $Y_1, \ldots, Y_n$ be a random sample of size $n$ from the $McG (a, b, c, \theta, \gamma)$ distribution and ${\boldsymbol\Theta}= (a, b, c, \theta, \gamma)$ be the unknown parameter vector. The log-likelihood function is given by
\begin{eqnarray}\label{eq.lik}
l({\boldsymbol\Theta}) &=& n \log (c \theta) -n \log (B(a/c,b)) + \gamma \sum _{i=1}^{n} y_i -\frac{\theta}{\gamma}‎‎ \sum_{i=1}^{n} (e^{\gamma y_i}-1)\nonumber‎\\
&&+‎ (a-1) \sum_{i=1}^{n} ‎\log (1-t_i) + (b-1) \sum_{i=1}^{n} \log (1-(1-t_i)^c),
\end{eqnarray}
where $t_i= \exp({-‎\frac{\theta}{\gamma} (e^{\gamma y_i}-1)‎})$. The maximum likelihood estimation (MLE) of ${\boldsymbol\Theta}$ is obtained by solving the nonlinear equations, $U({\boldsymbol\Theta})= (U_{a}({\boldsymbol\Theta}), U_{b}({\boldsymbol\Theta}), U_{c}({\boldsymbol\Theta}), U_{{\theta}}(\boldsymbol\Theta), U_{\gamma}({\boldsymbol\Theta}))^T= ‎\boldsymbol{0}‎$, where
\begin{eqnarray*}
U_{a}({\boldsymbol\Theta})&=&\frac{\partial l(\boldsymbol\Theta)}{\partial a}={n/c} [\psi (a/c+b) - \psi (a/c)]+ \sum_{i=1}^{n} \log (1-t_{i}),\\
U_{b}({\boldsymbol\Theta})&=&\frac{\partial l(\boldsymbol\Theta)}{\partial b}= n [\psi (a/c+b) -\psi(b)] + \sum_{i=1}^{n} \log(1-(1-t_{i}^{c})),\\
U_{c}({\boldsymbol\Theta})&=& \frac{\partial l(\boldsymbol\Theta)}{\partial c}= n/c - na/{c^2} [\psi (a/c+b) -\psi(a/c)] \\
&&- (b-1) \sum_{i=1}^{n} \frac{(1-t_{i}^{c}) \log(1-t_{i})}{1- (1-t_{i})^{c}}‎‎ ,\\
U_{{\theta}}(\boldsymbol\Theta)&=& \frac{\partial l(\boldsymbol\Theta)}{\partial \theta}= {n}/{\theta} - 1/{\gamma} \sum_{i=1}^{n} (e^{\gamma y_i}-1) + (a-1)/{\gamma} \sum_{i=1}^{n}\frac{t_i (e^{\gamma y_i}-1)}{1-t_i}\\
&&- c(b-1)/{\gamma} \sum_{i=1}^{n} ‎\frac{t_i (1-t_i)^{c-1} (e^{\gamma y_i}-1)}{1-(1-t_i)^c},\\
U_{\gamma}({\boldsymbol\Theta}) &=& \frac{\partial l(\boldsymbol\Theta)}{\partial \gamma}= \sum_{i=1}^{n} y_i + \theta / {\gamma ^2} \sum_{i=1}^{n} (e^{\gamma y_i} - \gamma y_i e^{\gamma y_i} -1) \\
&&+ \theta (a-1) / {\gamma ^2} \sum_{i=1}^{n}\frac{t_i (\gamma y_i e^{\gamma y_i}- e^{\gamma y_i}+1) }{1-t_i}‎\\
&&+ \theta (b-1) c/ {\gamma ^2} \sum_{i=1}^{n}\frac{t_i (1-t_i)^{c-1} (e^{\gamma y_i} - \gamma y_i e^{\gamma y_i} -1) }{1- (1-t_i)^c}‎.
\end{eqnarray*}

We need the observed information matrix for interval estimation and hypotheses tests on the model parameters. The $5‎\times 5$‎ Fisher information matrix, $J=J_n (\boldsymbol\Theta)$, is given by
\[J=- \left[ \begin{array}{ccccc}
J_{aa} & J_{ab} & J_{ac} & J_{a \theta} & J_{a \gamma} \\
J_{ba} & J_{bb} & J_{bc} & J_{b \theta} & J_{b \gamma} \\
J_{ca} & J_{cb} & J_{cc} & J_{c \theta} & J_{c \gamma} \\
J_{\theta a} & J_{\theta b} &J_{\theta c} & J_{\theta \theta} & J_{\theta \gamma} \\
J_{\gamma a} & J_{\gamma b} &J_{\gamma c} & J_{\gamma \theta} & J_{\gamma \gamma}
\end{array} \right],\]
where the expressions for the elements of $J$ are
 \begin{eqnarray*}
 J_{aa } &=& \frac{\partial^{2} l(\boldsymbol\Theta)}{\partial a^{2}}= \frac{n}{c^2} [‎\psi'‎(a/c+b) - ‎\psi'‎(a/c)],\\
 J_{ab}&=&\frac{\partial^{2} l(\boldsymbol\Theta)}{\partial a \partial b } =\frac{n}{c} \psi'‎(a/c+b),\\
 J_{ac}&=&\frac{\partial^{2} l(\boldsymbol\Theta)}{\partial a \partial c } = -\frac{na}{c^3} [‎\psi'(a/c+b) - ‎\psi'(a/c)],\\
 J_{a \theta}&=&\frac{\partial^{2} l(\boldsymbol\Theta)}{\partial a \partial \theta }=\frac{1}{\gamma}\sum_{i=1}^{n}\frac{{t_i} (e^{\gamma y_i}-1)}{1-t_i},\\
 J_{a \gamma}&=& \frac{\partial^{2} l(\boldsymbol\Theta)}{\partial a \partial \gamma} =\frac{\theta}{{\gamma}^2}\sum_{i=1}^{n} ‎\frac{t_i (\gamma y_i e^{\gamma y_i}-e^{\gamma y_i}+1)}{1-t_i},\\
 J_{bb}&=& \frac{\partial^{2} l(\boldsymbol\Theta)}{\partial b^2 } = n [‎\psi'(a/c+b) - ‎\psi'(b)],\\
 J_{bc}&=&\frac{\partial^{2} l(\boldsymbol\Theta)}{\partial b \partial c } = -\frac{na}{c^2} \psi'(a/c+b) - \sum_{i=1}^{n} ‎\frac{(1-t_i)^c \log (1-t_i)}{1-(1-t_i)^c},\\
 J_{b \theta}&=& \frac{\partial^{2} l(\boldsymbol\Theta)}{\partial b \partial \theta } =-\frac{c}{\gamma} \sum_{i=1}^{n} ‎\frac{ t_i (1-t_i)^{c-1} (e^{\gamma y_i}-1)}{1-(1-t_i)^c},\\
 J_{b\gamma}&=& \frac{\partial^{2} l(\boldsymbol\Theta)}{\partial b \partial \gamma } = - \frac{c \theta}{\gamma ^2}\sum_{i=1}^{n} ‎\frac{ t_i (1-t_i)^{c-1} (\gamma y_i e^{\gamma y_i}-e^{\gamma y_i}+1)}{1-(1-t_i)^c},\\
 J_{cc}&=& \frac{\partial^{2} l(\boldsymbol\Theta)}{\partial c^{2}} =- \frac{n}{c^2} + \frac{2na}{c^3} [\psi(a/c+b)  -\psi(a/c)] \\
&& + \frac{na^2}{c^4} [‎\psi'(a/c+b) - ‎\psi'(a/c)]- (b-1) \sum_{i=1}^{n} ‎\frac{(1-t_i)^c (\log (1-t_i))^2}{[1-(1-t_i)^c]^2},\\
 J_{c \theta}&=& \frac{\partial^{2} l(\boldsymbol\Theta)}{\partial c \partial \theta} = -\frac{b-1}{\gamma} \sum_{i=1}^{n} ‎\frac{t_i(1-t_i)^{c-1} (e^{\gamma y_i}-1) [c \log (1-t_i)+1- (1-t_i)^c]}{[1-(1-t_i)^c]^2},\\
J_{c \gamma}&=& \frac{\partial^{2} l(\boldsymbol\Theta)}{\partial c \partial \gamma}= -\frac{\theta(b-1)}{\gamma^2} \sum_{i=1}^{n} ‎\frac{t_i (\gamma y_i e^{\gamma y_i}-e^{\gamma y_i}+1) [c \log (1-t_i)+1- (1-t_i)^c]}{(1-t_i)^{1-c}[1-(1-t_i)^c]^2},\\
 J_{\theta \theta}&=& \frac{\partial^{2} l(\boldsymbol\Theta)}{\partial \theta^{2}} = -\frac{n}{\theta ^2} - \frac{a-1}{\gamma ^2} \sum_{i=1}^{n}
  ‎\frac{{t_i} (e^{\gamma y_i}-1)^2}{(1-t_i)^2} -\frac{c (b-1)}{\gamma ^2}\\
 && \times \sum_{i=1}^{n} ‎\frac{{t_i} (1-t_i)^{c-2} (e^{\gamma y_i}-1)^2 ‎\lbrace‎ { c t_i + (1-t_i)^c - 1} \rbrace‎‎}{(1-(1-t_i)^c)^2},\\
 J_{\theta \gamma}&=& \frac{\partial^{2} l(\boldsymbol\Theta)}{\partial \theta \partial \gamma}= ‎\frac{1}{\gamma ^2} \sum_{i=1}^{n} (e^{\gamma y_i} - \gamma y_i e^{\gamma y_i} -1)-‎\dfrac{c(b-1)}{\gamma ^3} \sum_{i=1}^{n} ‎\frac{t_i (e^{\gamma y_i} - \gamma y_i e^{\gamma y_i} -1)}{(1-t_i)^{2-c}(1-(1-t_i)^c)^2}\\
 && \times ‎\left[ (1-t_i)^c (\theta e^{\gamma y_i}+ t_i \gamma - \gamma - \theta) + c \theta t_i (e^{\gamma y_i}-1) + \gamma (1-t_i) + \theta (1-e^{\gamma y_i}) ‎\right]\\
 && + ‎\dfrac{(a-1)}{\gamma ^3} \sum_{i=1}^{n} ‎\dfrac{t_i (e^{\gamma y_i} - \gamma y_i e^{\gamma y_i} -1) (\theta e^{\gamma y_i}+ \gamma t_i - \gamma - \theta)}{(1-t_i)^2}‎‎‎‎‎‎,\\
 J_{\gamma\gamma}&=& \frac{\partial^{2} l(\boldsymbol\Theta)}{\partial \gamma^{2}}= ‎\dfrac{2 \theta}{\gamma ^3}‎ \sum_{i=1}^{n} (\gamma y_i e^{\gamma y_i} - e^{\gamma y_i} - \gamma ^2 y_{i}^2e^{\gamma y_i}/2 +1)\\
 && - ‎\dfrac{c(b-1)\theta ^2}{\gamma ^4} \sum_{i=1}^{n} ‎\dfrac{t_{i}^2 (1-t_i)^{c-2}(\gamma y_i e^{\gamma y_i}-e^{\gamma y_i}+1)^2 (2c(1-t_i)^c+(1-t_i)^c-c-1)}{(1-(1-t_i)^c)^2}\\
 && + \dfrac{c(b-1)\theta}{\gamma ^4}‎‎ \sum_{i=1}^{n} \dfrac{t_{i} (1-t_i)^{c-1}}{1-(1-t_i)^c} [ -\gamma ^3 y_{i}^2 e^{\gamma y_i} + 2 \gamma ^2 y_{i} e^{\gamma y_i}- 2 \gamma e^{\gamma y_i} + 2 \gamma + \theta \gamma ^2 y_{i}^2 e^{2 \gamma y_i}\\
 && + \theta (e^{\gamma y_i}-1)^2 - 2 \theta \gamma y_{i} e^{\gamma y_i} (e^{\gamma y_i}-1) ]\\
 && - ‎‎ \dfrac{(a-1)\theta}{\gamma ^4} \sum_{i=1}^{n} ‎\dfrac{t_i}{1-t_i} ‎\left[ -\gamma ^3 y_{i}^2 e^{\gamma y_i} + 2 \gamma ^2 y_{i} e^{\gamma y_i}-2 \gamma (e^{\gamma y_i}-1)+\theta (\gamma y_{i} e^{\gamma y_i}- e^{\gamma y_i}+1)^2 ‎\right]\\
 && - \dfrac{(a-1)\theta ^2}{\gamma ^4} \sum_{i=1}^{n}‎\dfrac{t_{i}^2 (\gamma y_{i} e^{\gamma y_i}- e^{\gamma y_i}+1)^2}{(1-t_i)^2}.
\end{eqnarray*}
Under conditions that are fulfilled for parameters in the interior of the parameter space but not on the boundary, asymptotically
$$\sqrt{n}(\hat{\boldsymbol\Theta}-\boldsymbol\Theta)\sim N_{5}(0,I(\boldsymbol\Theta)^{-1}),$$
where $I({\boldsymbol\Theta})$ is the expected information matrix. This asymptotic behavior is valid if $I(\boldsymbol\Theta)$ replaced by $J_n(\hat{\boldsymbol\Theta })$ , i.e., the observed information matrix evaluated at $\hat{\boldsymbol\Theta}$ \cite{co-hi-74}.

\section{Application of McG to two real data sets}
\label{sec.exa}
In this section, two real data sets are considered to illustrate that the McG model can be a good lifetime distribution comparing with main three submodels; BG, KumG and McE distributions. In both examples, we obtain the MLE and their corresponding standard errors (in parentheses) of the model parameters. The model selection is carried out using minus of log-likelihood function ($-\log  (L)$), Kolmogorov-Smirnov (K-S) statistic with its p-value, Akaike information criterion (AIC), Akaike information criterion corrected (AICC), Bayesian information criterion (BIC) and likelihood ratio test (LRT) with its p-value. Furthermore, we plot the histogram for each data set and the estimated pdf of the four models. Moreover, the plots of empirical cdf of the data sets and estimated cdf of four models are displayed.

\begin{example}
The data set have been obtained from \cite{Aarset-87} and represents the lifetimes of 50 devices. Also, it is analyzed by \cite{el-al-al-13} and \cite{ja-ta-al-14}. The results which are given in Table 1 indicate that the McE model is not suitable for this data set based on K-S statistic. The McG model has the lowest $-\log  (L)$, AIC and AICC values among all fitted models, but the BG model has the lowest BIC value among all other fitted models. A comparison of the proposed distribution with some of its submodels using p-value of LRT shows that the McG model yields a better fit
than the other three distributions to this real data set. It is also clear from Figure \ref{plot.ex1} that the McG distribution provides a better fit and therefore be one of the best models for this data set.
\end{example}

\begin{example}
The data set represents the strengths of 1.5 cm glass fibers, measured at the National Physical
Laboratory, England. Unfortunately, the units of measurement are not given in the paper. It is obtained from
\cite{sm-na-87}
and also analyzed by \cite{ba-sa-co-10}. The K-S statistic indicates that all fitted models are good candidates for this data set, but based on all other criteria: $-\log  (L)$, AIC, AICC, BIC and p-value of LRT, we can infer that the best model is the McG distribution. Similar  results are concluded from Figure \ref{plot.ex2}.
\end{example}

\begin{table}
\begin{center}
\caption{MLEs of the model parameters for the data of lifetimes of 50 devices, the corresponding SEs and the K-S, AIC, AICC, BIC and LRT statistics.}

\begin{tabular}{|l|c|c|c|c|} \hline
 & \multicolumn{4}{|c|}{Distribution} \\ \hline
 & BG & KumG & McE & McG \\ \hline
$\hat{a}$ & 0.2158 & 0.2374 & 0.8643 & 0.2619 \\
(s.e.) & (0.0392) & (0.0871) & (0.4621) & (0.0656) \\ \hline

$\hat{b}$ & 0.2467 & 0.6063 & 0.0706 & 0.0752 \\
(s.e.) & (0.0448) & (0.3780) & (0.0927) & (0.1029) \\ \hline

$\hat{c}$ & --- & 0.2374 & 2.7127 & 3.7652 \\
(s.e.) & --- & (0.0871) & (3.7438) & (0.9946) \\ \hline

$\hat{\theta }$ & 0.0003 & 0.0003 & 0.2826 & 0.0012 \\
(s.e.) & (0.0001) & (0.0001) & (0.3727) & (0.0001) \\ \hline

$\hat{\gamma }$ & 0.0882 & 0.0766 & --- & 0.0875 \\
(s.e.) & (0.0030) & (0.0282) & --- & (0.0001) \\ \hline\hline

$-\log  (L)$ & 220.6714 & 221.9666 & 237.8158 & 219.0041 \\
K-S & 0.1322 & 0.1367 & 0.1916 & 0.1216 \\
p-value (K-S) & 0.3456 & 0.3072 & 0.0507 & 0.4509 \\
AIC & 449.3437 & 451.9331 & 483.6316 & 448.0081 \\
AICC & 450.2326 & 452.8220 & 484.5205 & 449.3718 \\
BIC & 456.9918 & 459.5812 & 491.2797 & 457.5682 \\
LRT & 3.3356 & 5.9249 & 37.6233 & --- \\
p-value (LRT) & 0.06779 & 0.0149 & 0.0000 & --- \\ \hline
\end{tabular}

\caption{MLEs of the model parameters for the strengths of 1.5 cm glass fibers data, the corresponding SEs and the K-S, AIC, AICC, BIC and LRT statistics.}

\begin{tabular}{|l|c|c|c|c|} \hline
 & \multicolumn{4}{|c|}{Distribution} \\ \hline
 & BG & KumG & McE & McG \\ \hline
$\hat{a}$ & 1.6907 & 1.8946 & 9.3276 & 0.7940 \\
(s.e.) & (0.8664) & (1.2350) & (2.6891) & (0.2355) \\ \hline

$\hat{b}$ & 27.7434 & 4.2814 & 93.4655 & 0.1248 \\
(s.e.) & (87.8726) & (15.7612) & (109.2042) & (0.1786) \\ \hline

$\hat{c}$ & --- & 1.8946 & 22.6124 & 192.1704 \\
(s.e.) & --- & (1.2350) & (14.0674) & (307.3698) \\ \hline

$\hat{\theta }$ & 0.0020 & 0.0309 & 0.9227 & 0.0009 \\
(s.e.) & (0.0632) & (0.0509) & (0.3443) & (0.0001) \\ \hline

$\hat{\gamma }$ & 2.7156 & 2.3019 & --- & 5.2013 \\
(s.e.) & (0.8448) & (1.6517) & --- & (0.4018) \\ \hline \hline

$-\log (L)$ & 14.2158 & 14.0305 & 15.5995 & 11.4208 \\
K-S & 0.1324 & 0.1313 & 0.1466 & 0.1159 \\

p-value (K-S) & 0.2186 & 0.2271 & 0.1335 & 0.3651 \\
AIC & 36.4317 & 36.0610 & 37.2569 & 32.8417 \\
AICC & 37.1213 & 36.7506 & 37.9466 & 33.8943 \\
BIC & 45.0042 & 44.6335 & 45.8295 & 43.5573 \\
LRT & 5.5900 & 5.2193 & 8.3572 & --- \\
p-value (LRT) & 0.0181 & 0.0223 & 0.0038 & --- \\ \hline
\end{tabular}

\end{center}
\end{table}

\begin{figure}
\centering
\includegraphics[scale=0.53]{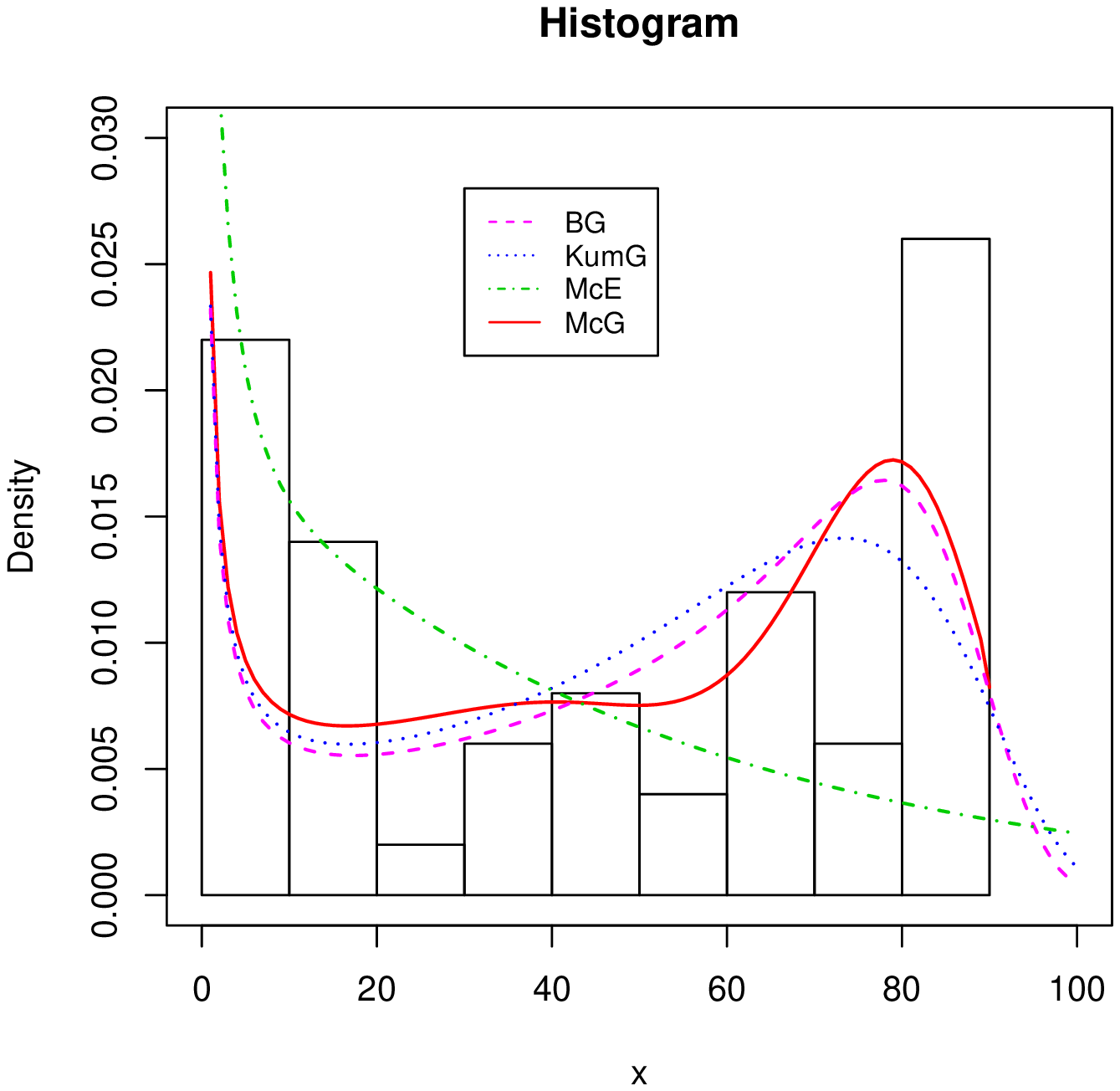}
\includegraphics[scale=0.53]{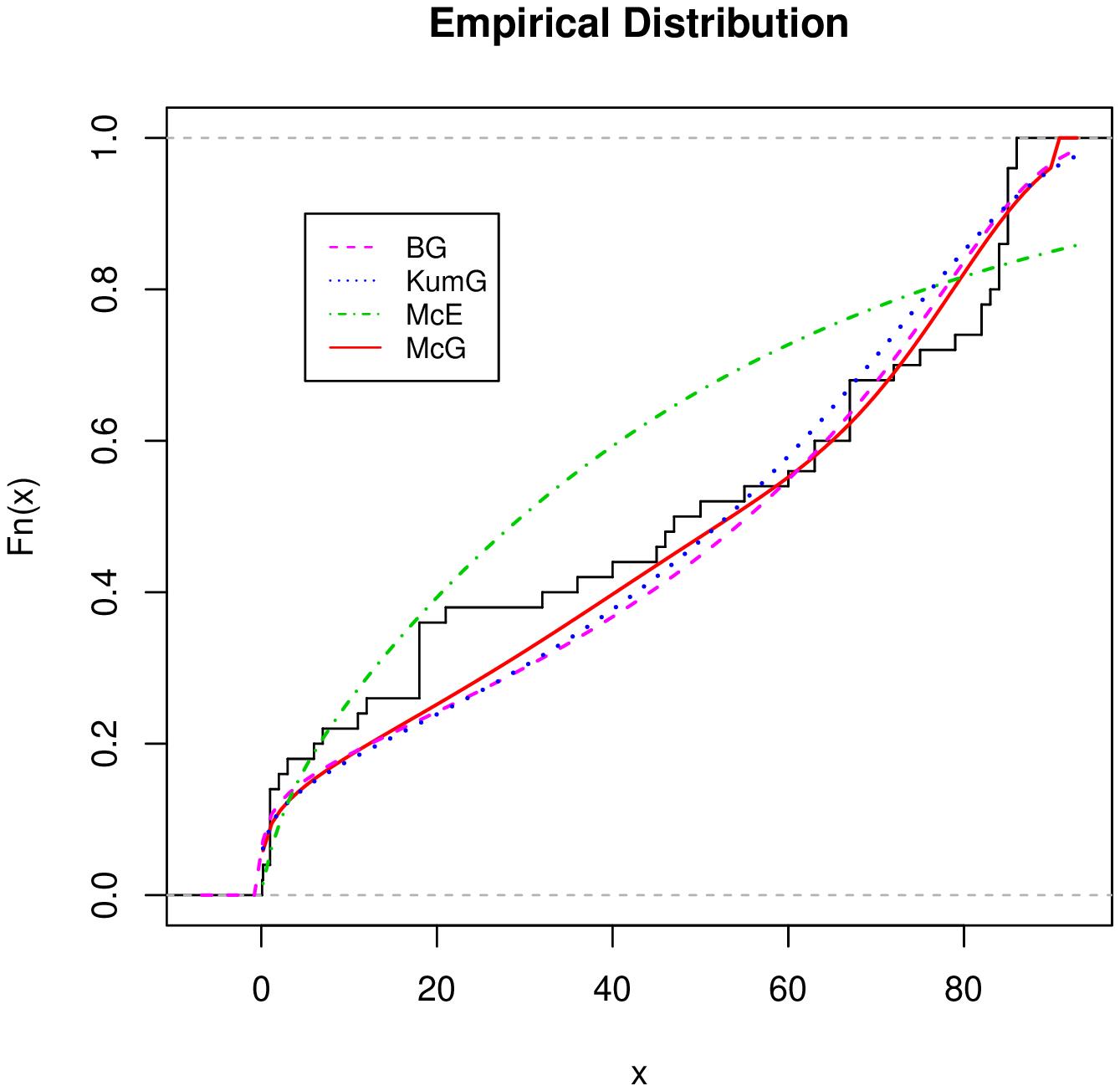}
\caption[]{Plots of the estimated pdfs and cdfs BG, KumG, McE and McG models using data of lifetimes of 50 devices. }\label{plot.ex1}

\includegraphics[scale=0.55]{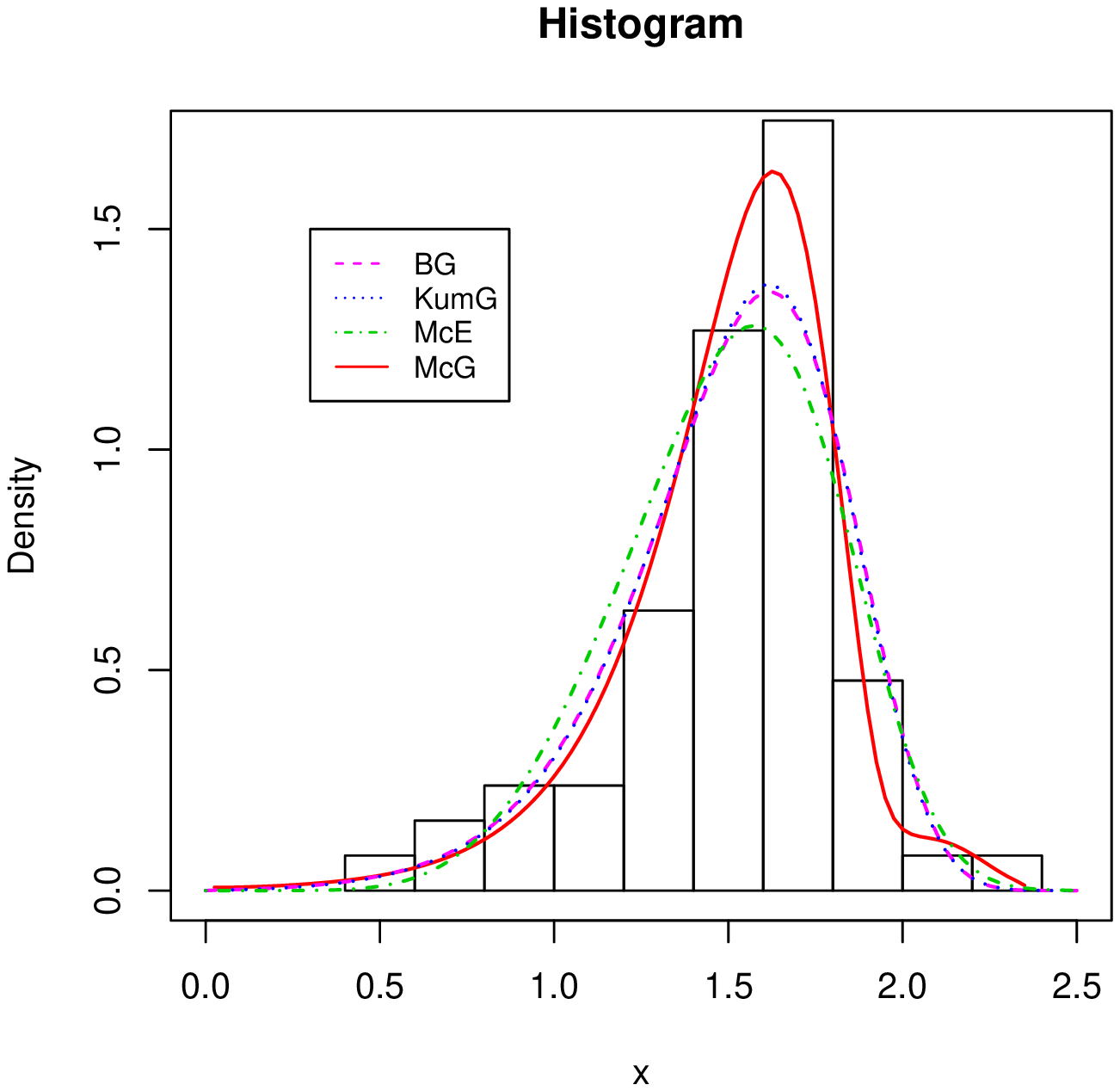}
\includegraphics[scale=0.55]{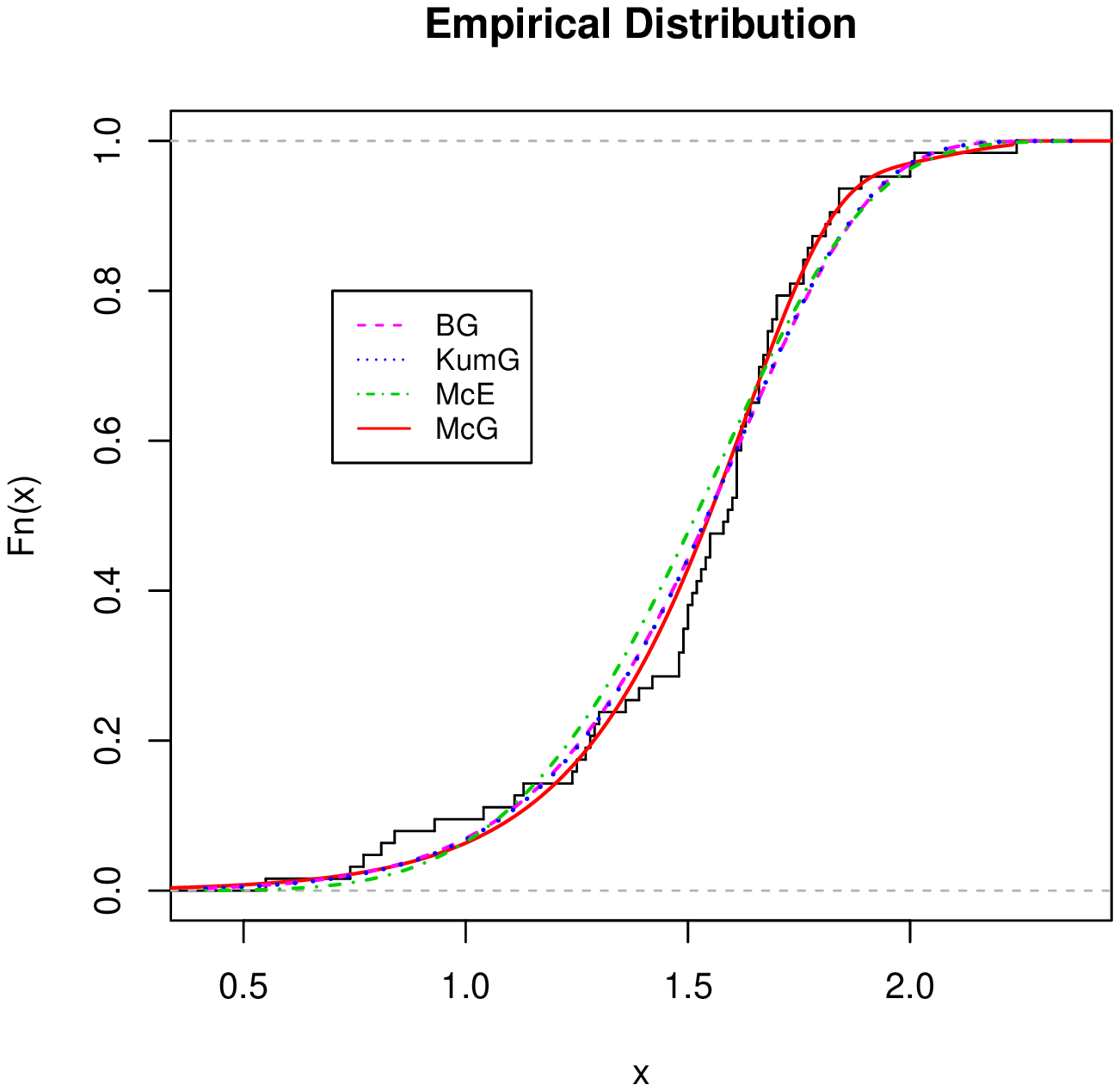}
\caption[]{Plots of the estimated pdfs and cdfs BG, KumG, McE and McG models using the strengths of 1.5 cm glass fibers data.}\label{plot.ex2}
\end{figure}

\newpage
\bibliographystyle{apa}

\end{document}